\documentclass[letterpaper,11pt,twoside]{article}
\usepackage{amsmath,amssymb,theorem,xspace}
\usepackage[dvips]{graphicx}
\usepackage{ifthen,calc,epic}
\usepackage[usenames,dvipsnames]{color}
\usepackage{hyperref,breakurl,mathptmx} 

\advance\topmargin-\headheight           
\advance\topmargin-\headsep              
\advance\evensidemargin  .295cm          
\advance\oddsidemargin   .295cm          
\theoremheaderfont{\scshape}
\theorembodyfont{\normalfont\slshape}
\theoremstyle{plain}
\newtheorem{theorem}{Theorem}
\newtheorem{proposition}{Proposition}
\newtheorem{lemma}{Lemma}
\newtheorem{conjecture}{Conjecture}
\newtheorem{corollary}{Corollary}

\theorembodyfont{\normalfont}

\newenvironment{proof}{\begin{trivlist}\item{}\normalfont\textit{Proof. }}{\hfill$\square$\end{trivlist}}

\newdimen\proofrulebreadth \proofrulebreadth=.05em
\newdimen\proofdotseparation \proofdotseparation=1.25ex
\newdimen\proofrulebaseline \proofrulebaseline=2ex
\newcount\proofdotnumber \proofdotnumber=3
\let\then\relax
\def\hfi{\hskip0pt plus.0001fil}
\mathchardef\squigto="3A3B
\newif\ifinsideprooftree\insideprooftreefalse
\newif\ifonleftofproofrule\onleftofproofrulefalse
\newif\ifproofdots\proofdotsfalse
\newif\ifdoubleproof\doubleprooffalse
\let\wereinproofbit\relax
\newdimen\shortenproofleft
\newdimen\shortenproofright
\newdimen\proofbelowshift
\newbox\proofabove
\newbox\proofbelow
\newbox\proofrulename
\def\shiftproofbelow{\let\next\relax\afterassignment\setshiftproofbelow\dimen0 }
\def\shiftproofbelowneg{\def\next{\multiply\dimen0 by-1 }%
\afterassignment\setshiftproofbelow\dimen0 }
\def\setshiftproofbelow{\next\proofbelowshift=\dimen0 }
\def\setproofrulebreadth{\proofrulebreadth}
\def\prooftree{
\ifnum  \lastpenalty=1
\then   \unpenalty
\else   \onleftofproofrulefalse
\fi
\ifonleftofproofrule
\else   \ifinsideprooftree
        \then   \hskip.5em plus1fil
        \fi
\fi
\bgroup
\setbox\proofbelow=\hbox{}\setbox\proofrulename=\hbox{}%
\let\justifies\proofover\let\leadsto\proofoverdots\let\Justifies\proofoverdbl
\let\using\proofusing\let\[\prooftree
\ifinsideprooftree\let\]\endprooftree\fi
\proofdotsfalse\doubleprooffalse
\let\thickness\setproofrulebreadth
\let\shiftright\shiftproofbelow \let\shift\shiftproofbelow
\let\shiftleft\shiftproofbelowneg
\let\ifwasinsideprooftree\ifinsideprooftree
\insideprooftreetrue
\setbox\proofabove=\hbox\bgroup$\displaystyle 
\let\wereinproofbit\prooftree
\shortenproofleft=0pt \shortenproofright=0pt \proofbelowshift=0pt
\onleftofproofruletrue\penalty1
}
\def\eproofbit{
\ifx    \wereinproofbit\prooftree
\then   \ifcase \lastpenalty
        \then   \shortenproofright=0pt  
        \or     \unpenalty\hfil         
        \or     \unpenalty\unskip       
        \else   \shortenproofright=0pt  
        \fi
\fi
\global\dimen0=\shortenproofleft
\global\dimen1=\shortenproofright
\global\dimen2=\proofrulebreadth
\global\dimen3=\proofbelowshift
\global\dimen4=\proofdotseparation
\global\count255=\proofdotnumber
$\egroup  
\shortenproofleft=\dimen0
\shortenproofright=\dimen1
\proofrulebreadth=\dimen2
\proofbelowshift=\dimen3
\proofdotseparation=\dimen4
\proofdotnumber=\count255
}
\def\proofover{
\eproofbit 
\setbox\proofbelow=\hbox\bgroup 
\let\wereinproofbit\proofover
$\displaystyle
}%
\def\proofoverdbl{
\eproofbit 
\doubleprooftrue
\setbox\proofbelow=\hbox\bgroup 
\let\wereinproofbit\proofoverdbl
$\displaystyle
}%
\def\proofoverdots{
\eproofbit 
\proofdotstrue
\setbox\proofbelow=\hbox\bgroup 
\let\wereinproofbit\proofoverdots
$\displaystyle
}%
\def\proofusing{
\eproofbit 
\setbox\proofrulename=\hbox\bgroup 
\let\wereinproofbit\proofusing
\kern0.3em$
}
\def\endprooftree{
\eproofbit 
  \dimen5 =0pt
\dimen0=\wd\proofabove \advance\dimen0-\shortenproofleft
\advance\dimen0-\shortenproofright
\dimen1=.5\dimen0 \advance\dimen1-.5\wd\proofbelow
\dimen4=\dimen1
\advance\dimen1\proofbelowshift \advance\dimen4-\proofbelowshift
\ifdim  \dimen1<0pt
\then   \advance\shortenproofleft\dimen1
        \advance\dimen0-\dimen1
        \dimen1=0pt
        \ifdim  \shortenproofleft<0pt
        \then   \setbox\proofabove=\hbox{%
                        \kern-\shortenproofleft\unhbox\proofabove}%
                \shortenproofleft=0pt
        \fi
\fi
\ifdim  \dimen4<0pt
\then   \advance\shortenproofright\dimen4
        \advance\dimen0-\dimen4
        \dimen4=0pt
\fi
\ifdim  \shortenproofright<\wd\proofrulename
\then   \shortenproofright=\wd\proofrulename
\fi
\dimen2=\shortenproofleft \advance\dimen2 by\dimen1
\dimen3=\shortenproofright\advance\dimen3 by\dimen4
\ifproofdots
\then
        \dimen6=\shortenproofleft \advance\dimen6 .5\dimen0
        \setbox1=\vbox to\proofdotseparation{\vss\hbox{$\cdot$}\vss}%
        \setbox0=\hbox{%
                \advance\dimen6-.5\wd1
                \kern\dimen6
                $\vcenter to\proofdotnumber\proofdotseparation
                        {\leaders\box1\vfill}$%
                \unhbox\proofrulename}%
\else   \dimen6=\fontdimen22\the\textfont2 
        \dimen7=\dimen6
        \advance\dimen6by.5\proofrulebreadth
        \advance\dimen7by-.5\proofrulebreadth
        \setbox0=\hbox{%
                \kern\shortenproofleft
                \ifdoubleproof
                \then   \hbox to\dimen0{%
                        $\mathsurround0pt\mathord=\mkern-6mu%
                        \cleaders\hbox{$\mkern-2mu=\mkern-2mu$}\hfill
                        \mkern-6mu\mathord=$}%
                \else   \vrule height\dimen6 depth-\dimen7 width\dimen0
                \fi
                \unhbox\proofrulename}%
        \ht0=\dimen6 \dp0=-\dimen7
\fi
\let\doll\relax
\ifwasinsideprooftree
\then   \let\VBOX\vbox
\else   \ifmmode\else$\let\doll=$\fi
        \let\VBOX\vcenter
\fi
\VBOX   {\baselineskip\proofrulebaseline \lineskip.2ex
        \expandafter\lineskiplimit\ifproofdots0ex\else-0.6ex\fi
        \hbox   spread\dimen5   {\hfi\unhbox\proofabove\hfi}%
        \hbox{\box0}%
        \hbox   {\kern\dimen2 \box\proofbelow}}\doll%
\global\dimen2=\dimen2
\global\dimen3=\dimen3
\egroup 
\ifonleftofproofrule
\then   \shortenproofleft=\dimen2
\fi
\shortenproofright=\dimen3
\onleftofproofrulefalse
\ifinsideprooftree
\then   \hskip.5em plus 1fil \penalty2
\fi
}

\makeatletter
\def\section{\@startsection{section}{1}{0pt}{-3.25ex plus -1ex minus 
-.2ex}{1.5ex plus .2ex minus .3ex}{\normalfont\large\bf}}
\def\subsection{\@startsection {subsection}{2}{0pt}{-2ex plus -1ex minus 
   -.2ex}{1.5ex plus .2ex minus .3ex}{\normalfont\normalsize\bf}}
\makeatother

\newlength{\tw}
\newcommand{\cutoption}[1]{}
\newcommand{\RvG}[1]{#1} 
\newcommand{\plat}[1]{\raisebox{0pt}[0pt][0pt]{#1}}     
	
\newcommand{\Green}{OliveGreen}
\newcommand{\Blue}{Blue}

\newcommand{\Red}{Red}

\newcommand{\Brown}{Brown}

\newcommand{\linkcolor}{Black}
\newcommand{\colorlinks}{\color{\linkcolor}}

\newcounter{l}
\newcommand{\rln}[3]{\colorlinks\setcounter{l}{#2-#1}\put(#1,#3){\line(1,0){\value{l}}}}
\newcommand{\lln}[3]{\colorlinks\setcounter{l}{#1-#2}\put(#2,#3){\line(1,0){\value{l}}}}
\newcommand{\uln}[3]{\colorlinks\setcounter{l}{#3-#2}\put(#1,#2){\line(0,1){\value{l}}}}
\newcommand{\dln}[3]{\colorlinks\setcounter{l}{#2-#3}\put(#1,#2){\line(0,-1){\value{l}}}}

\newcommand{\hln}[3]{\ifthenelse{#2 > #1}
{\rln{#1}{#2}{#3}}{\lln{#1}{#2}{#3}}}
\newcommand{\vln}[3]{\ifthenelse{#3 > #2}
{\uln{#1}{#2}{#3}}{\dln{#1}{#2}{#3}}}
\newcommand{\link}[4]{
\vln{#1}{#3}{#4}\vln{#2}{#3}{#4}\hln{#1}{#2}{#4}}

\newlength{\boxht}
\newcommand{\obox}[2]{\raisebox{#1}{\makebox[0pt][l]{#2}}}
\newcommand{\oboxlink}[5]{\obox{#1}{\begin{picture}(0,#5)\link{#2}{#3}{#4}{#5}\end{picture}}}
\newcommand{\olinkexplicit}[4]{\oboxlink{\boxht}{#1}{#2}{#3}{#4}}
\newlength{\boxdp}\newlength{\uboxlocalvar}
\newcommand{\ubox}[2]{\settoheight{\uboxlocalvar}{#2}\addtolength{\uboxlocalvar}{#1}%
\raisebox{-\uboxlocalvar}{\makebox[0pt][l]{#2}}}
\newcommand{\uboxlink}[5]{\ubox{#1}{\begin{picture}(0,#5)(0,-#5)\link{#2}{#3}{-#4}{-#5}\end{picture}}}
\newcommand{\ulinkexplicit}[4]{\uboxlink{\boxdp}{#1}{#2}{#3}{#4}}

\newcommand{\olinkbasey}{0}\newcommand{\ulinkbasey}{0}
\newcommand{\olinkspecy}[3]{\olinkexplicit{#1}{#2}{\olinkbasey}{#3}} 
\newcommand{\ulinkspecy}[3]{\ulinkexplicit{#1}{#2}{\ulinkbasey}{#3}} 

\newcommand{\olinky}{4}\newcommand{\ulinky}{\olinky}
\newcommand{\oolinky}{7}\newcommand{\uulinky}{\oolinky}

\newcommand{\olink}[2]{\olinkspecy{#1}{#2}{\olinky}}
\newcommand{\oolink}[2]{\olinkspecy{#1}{#2}{\oolinky}}

\newcommand{\ulink}[2]{\ulinkspecy{#1}{#2}{\ulinky}}
\newcommand{\uulink}[2]{\ulinkspecy{#1}{#2}{\uulinky}}

\newcommand{\centredlinks}[2]{
\settoheight{\boxht}{\ensuremath{#1}}\settodepth{\boxdp}{\ensuremath{#1}}#2\ensuremath{#1}}
\newlength{\updepthboxvar}
\newcommand{\updepthbox}[1]{\settodepth{\updepthboxvar}{#1}\raisebox{\updepthboxvar}{#1}}
\newcommand{\usualdepth}{.1ex}
\newcommand{\shiftdown}[1]{\raisebox{-\usualdepth}{#1}}
\newlength{\fattenheight}\newlength{\fattendepth}
\newcommand{\fattenbox}[1]{\settoheight{\fattenheight}{#1}\addtolength{\fattenheight}{.5ex}%
\settodepth{\fattendepth}{#1}\addtolength{\fattendepth}{0ex}\raisebox{0pt}[\fattenheight][\fattendepth]{#1}}
\newcommand{\links}[2]{\fattenbox{\shiftdown{\updepthbox{\centredlinks{#1}{#2}}}}}
\newlength{\commadepth}

\newcommand{\defn}[1]{{\textit{\textbf{#1}}}}

\newcommand{\ie}{\emph{i.e.}}

\newcommand{\eg}{\emph{e.g.}}

\newcommand{\cf}{\emph{cf.}}
\newcommand{\figurerule}{\vspace*{1.8ex}\hrule}

\newcommand{\filledbox}{\rule{1.2ex}{1.2ex}}

\renewcommand{\perp}{^\bot}
\newcommand{\with}{\&}
\newcommand{\plus}{\oplus}
\newcommand{\tensor}{\otimes}
\newcommand{\tightplus}{\!\oplus\!}
\newcommand{\tighttensor}{\!\otimes\!}
\newlength{\parrdp}\newlength{\parrht}
\newcommand{\parr}{\raisebox{-\parrdp}{\raisebox{\parrht}{\rotatebox{180}{$\&$}}}}
\newlength{\smallparrdp}\newlength{\smallparrht}
\newcommand{\smallparr}{\mkern1mu\raisebox{-.1\smallparrdp}{\raisebox{\smallparrht}{\rotatebox{180}{\small$\&$}}}\mkern1mu}
\newlength{\footnoteparrdp}\newlength{\footnoteparrht}
\newcommand{\footnoteparr}{\mkern1mu\raisebox{-\footnoteparrdp}{\raisebox{\footnoteparrht}{\rotatebox{180}{\footnotesize$\&$}}}\mkern1mu}
\newlength{\tinyparrdp}\newlength{\tinyparrht}
\newcommand{\tinyparr}{\raisebox{-\tinyparrdp}{\raisebox{\tinyparrht}{\rotatebox{180}{\tiny$\&$}}}}
\newcommand{\cutsymbol}{\ast}
\newcommand{\cut}{\mkern-1mu\cutsymbol\mkern-1mu}

\newcommand{\axlabel}{\mathsf{ax}}
\newcommand{\parlabel}{\parr}
\newcommand{\smallparlabel}{\smallparr}
\newcommand{\tensorlabel}{\tensor}
\newcommand{\withlabel}{\with}

\newcommand{\leftpluslabel}{\oplus_1}
\newcommand{\rightpluslabel}{\oplus_2}
\newcommand{\truecutlabel}{\mathsf{cut}}
\newcommand{\cutlabel}{\,\mbox{\normalsize$\cutsymbol$}}
\newcommand{\ipluslabel}{\oplus_i}
\newcommand{\jpluslabel}{\oplus_j}
\newcommand{\mixlabel}{\mathsf{mix}}

\newcommand{\axiomrule}[1]{\[\;\justifies {#1}\;\]}
\newcommand{\piproof}[2]{\[\;\Pi_{#1}\justifies{#2}\]}
\newcommand{\convparr}{\tinyparr}
\newcommand{\comm}[2]{\mathsf{C}^{_{#1}}_{^{#2}}}
\newcommand{\conv}[2]{\;\begin{array}{c}\comm{#1}{#2}\\\longrightarrow\\\longleftarrow\\\comm{#2}{#1}\end{array}\;}
\newcommand{\diagconv}[1]{\;\begin{array}{c}\comm{#1}{#1}\\\longleftrightarrow\end{array}\;}

\newcommand{\commpair}[2]{\comm{#1}{#2}\text{/}\comm{#2}{#1}}

\newcommand{\jumpgraph}[1]{\mathcal{G}_{#1}}
\newcommand{\fullgraph}{\jumpgraph{\theta}}

\newcommand{\judgement}{\raisebox{.15ex}{$\mkern9mu\rhd\mkern8mu$}}

\newcommand{\mallcut}{MALL$^{\mkern-5mu\cutlabel}$}

\newcommand{\seqrel}{\begin{picture}(24,4)\put(10.5,3){\makebox(0,0){\small$\mathsf{s}$}\
}\put(4,3){\vector(1,0){16}}\end{picture}}

\makeatletter
\def\@listi{\leftmargin\leftmargini
            \parsep 2.5\p@ \@plus1.5\p@ \@minus\p@
            \topsep 5\p@   \@plus2\p@ \@minus5\p@
            \itemsep2.5\p@ \@plus1.5\p@ \@minus\p@}
\let\@listI\@listi
\@listi
\def\@listii {\leftmargin\leftmarginii
              \labelwidth\leftmarginii
              \advance\labelwidth-\labelsep
              \topsep    1\p@ \@plus\p@ \@minus\p@
              \parsep    1\p@   \@plus\p@  \@minus\p@
              \itemsep   \parsep}
\def\@listiii{\leftmargin\leftmarginiii
              \labelwidth\leftmarginiii
              \advance\labelwidth-\labelsep
              \topsep    \z@
              \parsep    \z@
              \itemsep   \topsep}
\makeatother
\newcommand{\A}{A_1}
\newcommand{\B}{A_2}
\newcommand{\X}{B_1}
\newcommand{\Y}{B_2}

\newcommand{\convVgap}{4ex}
\newcommand{\convHgap}{\hspace{8ex}}

\newcommand{\cutfreeMALL}{MALL$^{\mkern-5mu-}$\xspace}
\newcommand{\fullMALL}{MALL\xspace}

\title{\fullMALL proof nets identify proofs\\ modulo rule commutation}
\author{%
   \sc Rob van Glabbeek \\ \small
	NICTA\!\thanks{NICTA is funded by the Australian 
        Government through the Department of Communications and 
        the Australian Research Council through the ICT Centre of 
        Excellence Program.}\;/Data61, CSIRO \& UNSW,\, Sydney
   \and 
   \sc Dominic Hughes \\ \small Stanford University and U.C.\
      Berkeley\thanks{This research was conducted primarily whilst a
        Visiting Scholar at Stanford in the Mathematics and Computer
        Science departments, and completed as a Visiting Scholar in
        the Berkeley Logic Group. I gratefully acknowledge my
        respective hosts, Sol Feferman, Vaughan Pratt, and Wes
        Holliday.}}
\date{}

\begin{document}
\vspace{-3ex}\maketitle

\vspace{-3ex}\begin{quotation}\small\noindent
We show that the proof nets introduced in \cite{HG03,HvG} for \fullMALL
(Multiplicative Additive Linear Logic, without units) identify cut-free proofs modulo rule
commutation: two cut-free proofs translate to the same proof net \emph{if and
only if} one can be obtained from the other by a succession of rule commutations.
This result holds with and without the mix rule, and we extend it with cut.
\end{quotation}

\section{Introduction}

The proof nets for \fullMALL (Multiplicative Additive Linear Logic
\cite{Gir87}, without units) introduced in \cite{HG03,HvG} solved numerous issues with
monomial proof nets \cite{Gir96}, for example:
\begin{itemize}
\item There is a simple (deterministic) translation function from
  cut-free proofs to proof nets.
\vspace{-2pt}
\item Cut elimination is simply defined and strongly normalising.
\vspace{-2pt}
\item Proof nets form a semi (\ie, unit-free) star-autonomous category with 
(co)products.
\end{itemize}  
A proof net is a set of \defn{linkings} on a sequent.  Each linking is
a set of \defn{links} between complementary formula leaves (literal
occurrences).  Figure~\ref{intro-translation-eg}
\begin{figure}[tb]
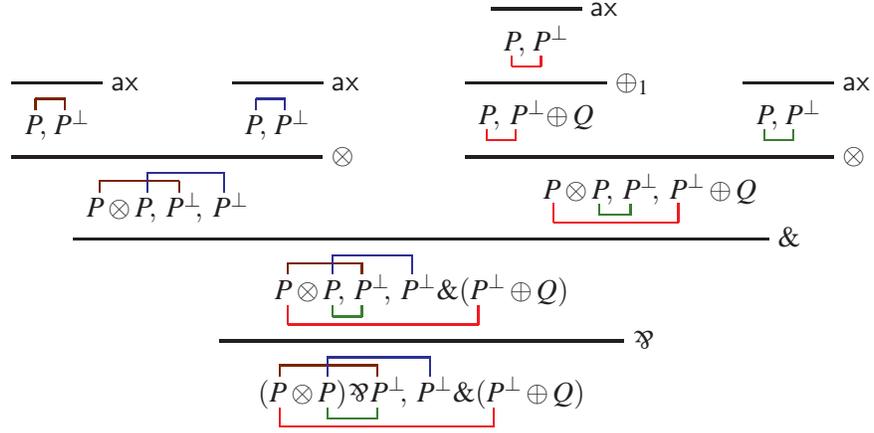

\begin{center}
\newcommand{\PPperpolink}{\links{P,\:P\perp}{\olink{4}{15}}}
\newcommand{\PPperpulink}{\links{P,\:P\perp}{\ulink{3}{14}}}
\newcommand{\tensorrulegap}{\hspace{6ex}}
\newcommand{\axrulegap}{\hspace{5ex}}
\prooftree
\thickness=0.08em
\[
  \[ \[ \justifies
  \;\;\renewcommand{\linkcolor}{\Brown}\PPperpolink\;\; \using
  \axlabel \] \axrulegap \[ \justifies
  \;\;\renewcommand{\linkcolor}{\Blue}\PPperpolink\;\; \using \axlabel
  \] \justifies \;\; \links{P\tensor
  P,\;P\perp\!\!,\;P\perp}{\renewcommand{\linkcolor}{\Brown}\olink{5}{35}\renewcommand{\linkcolor}{\Blue}\oolink{23}{52}}
  \;\; \using\tensorlabel \] \tensorrulegap \[ \[ \[ \justifies
  \;\;\renewcommand{\linkcolor}{\Red}\PPperpulink \;\; \using \axlabel
  \] \justifies \;\; \links{P,\;P\perp\!\plus
  Q}{\renewcommand{\linkcolor}{\Red}\ulink{3}{14}} \;\; \using
  \leftpluslabel \] \axrulegap \[ \justifies
  \;\;\renewcommand{\linkcolor}{\Green}\PPperpulink\;\; \using
  \axlabel \] \justifies \;\; \links{P\tensor
  P,\;P\perp\!\!,\;P\perp\plus
  Q}{\renewcommand{\linkcolor}{\Green}\ulink{21}{33}\renewcommand{\linkcolor}{\Red}\uulink{4}{51}}
  \;\; \using \tensorlabel \] \justifies \links{P\tensor
  P,\;P\perp\!\!,\;P\perp\with (P\perp\plus Q)}
  {\renewcommand{\linkcolor}{\Brown}\olink{5}{33}\renewcommand{\linkcolor}{\Blue}\oolink{22}{52}\renewcommand{\linkcolor}{\Green}\ulink{22}{33}\renewcommand{\linkcolor}{\Red}\uulink{5}{77}}
  \using \withlabel
\]
\justifies 
  \;\;\;\;\;\;
  \links{(P\tensor P)\smallparr P\perp\!\!,\;P\perp\with (P\perp\plus Q)}
    {\renewcommand{\linkcolor}{\Brown}\olink{8}{45}\renewcommand{\linkcolor}{\Blue}\oolink{26}{65}\renewcommand{\linkcolor}{\Green}\ulink{26}{45}\renewcommand{\linkcolor}{\Red}\uulink{8}{89}}
  \;\;\;\;\;\;
  \using \smallparlabel
\endprooftree
\end{center}
\caption{Example of the inductive
translation of a \fullMALL proof into a proof
net.\label{intro-translation-eg} The concluding proof net has two
linkings, one drawn above the sequent, the other below.  Each has two
links.  The proof nets further up in the derivation have one or two
linkings, correspondingly above/below the sequent.}%
\figurerule
\end{figure}%
illustrates the translation of a proof into a proof net.

In this paper we prove that the translation precisely captures
proofs modulo rule commutation:
two proofs translate to the same proof net \emph{if and only if} one
can be obtained from the other by a succession of rule commutations.
A \defn{rule commutation} is a
transposition of adjacent rules that preserves subproofs immediately above,
with possible duplication/identification, for example
\begin{displaymath}
\begin{array}{ccc}
\begin{prooftree}\thickness=.08em
\axiomrule{P\perp\!\!,P}
\quad
\[ \axiomrule{Q,Q\perp} \justifies Q,\;R\plus Q\perp\using\rightpluslabel \]
\justifies P\perp\!\!,\;P\tensor Q,\;R\plus Q\perp \using\tensorlabel
\end{prooftree}
&\quad\quad\longrightarrow\quad\quad&
\begin{prooftree}\thickness=.08em
\[\axiomrule{P\perp\!\!,P}\quad\quad \axiomrule{Q,Q\perp}\justifies
  P\perp\!\!,\;P\tensor Q\;,Q\perp\!\using\tensorlabel \]
\justifies   P\perp\!\!,\;P\tensor Q,\;R\plus Q\perp\using\rightpluslabel
\end{prooftree}
\end{array}
\end{displaymath}
in which the lower $\tensor$-rule commutes over the $\plus$-rule, or
\begin{displaymath}
\newcommand{\leftsubpr}{{\color{\Blue}\[\axiomrule{P\perp\!\!,P}\hspace{2ex} \justifies P\perp\!\!,P\plus R \using\leftpluslabel \]}}
\newcommand{\bigsub}{\[ \leftsubpr\axiomrule{Q\perp\!\!,Q} \justifies P\perp\!\!,(P\plus R)\tensor Q\perp\!\!,Q \using \tensorlabel \]}
\begin{array}{ccc}
\begin{prooftree}\thickness=.08em
\leftsubpr
\[ \axiomrule{Q\perp\!\!,Q} \axiomrule{Q\perp\!\!,Q}
   \justifies 
   Q\perp\!\!,\;Q\with Q\using\withlabel \]
\justifies P\perp\!\!,\;(P\plus R)\tensor Q\perp\!\!,\;Q\with Q \using\tensorlabel
\end{prooftree}
&\hspace{2.5ex}\longrightarrow\hspace{2.5ex}&
\begin{prooftree}\thickness=.08em
\bigsub
\bigsub
\justifies
P\perp\!\!,\;P\tensor Q\perp\!\!,\;Q\with Q \using \withlabel
\end{prooftree}
\end{array}
\end{displaymath}
illustrating duplication 
(of the $\tensor$-rule and subproof\hspace{.8ex} {\footnotesize\color{\Blue}\begin{prooftree}
\raisebox{-1.8ex}[0ex][0ex]{\begin{prooftree}\justifies\raisebox{.8ex}[0ex][0ex]{\footnotesize$P\perp\!\!,P$}\end{prooftree}} 
\justifies \raisebox{.8ex}[0ex][0ex]{$P\perp\!\!,P\plus R$} \using \leftpluslabel	  
\end{prooftree}}\,)
as the $\tensor$-rule commutes over the
$\with$-rule.

\section{Cut-free MALL}\label{mall-sec}
Let \cutfreeMALL denote cut-free multiplicative-additive linear logic without units
\cite{Gir87}.\footnote{We treat
cut in Section~\ref{cut-sec}.}
Formulas are built from literals (propositional
variables $P,Q,\ldots$ and their negations $P\perp$, $Q\perp,\ldots$)
by the binary connectives \defn{tensor}~$\otimes$, \defn{par}~$\parr$,
\defn{with}~$\with$ and \defn{plus}~$\oplus$.  
Negation\label{negation} $(-)\perp$ extends to arbitrary formulas
with  $P\perp{}\perp=P$ on propositional variables and
de Morgan duality: $(A\tensor B)\perp=A\perp\parr B\perp$, $(A\parr
B)\perp=A\perp\tensor B\perp$, $(A\plus B)\perp=A\perp\with B\perp$,
and $(A\with B)\perp=A\perp\plus B\perp$.  
We identify a formula with its parse tree, labelled with literals on
leaves and connectives on internal vertices.
A \defn{sequent} is a disjoint union of formulas.  Thus a
sequent is a labelled forest.  We write comma for
disjoint union.  For example, 
\newcommand{\simplemllsequent}{P\perp\!\!,\;\,(P\tensor P\perp)\parr P}%
$$\simplemllsequent$$ is the labelled forest
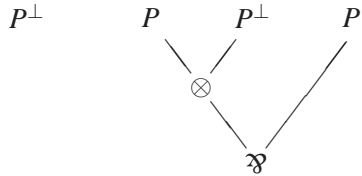
\begin{center}
\setlength{\unitlength}{1.9pt}
\begin{picture}(140,37)(-15,-18)
\newcommand{\rtwo}{15}
\newcommand{\rperp}{41}
\newcommand{\rtensor}{51}
\newcommand{\ess}{60}
\newcommand{\spar}{62}
\newcommand{\sperp}{81}
\newcommand{\toprow}{14}
\newcommand{\midrow}{0}
\newcommand{\botrow}{-15}
\newcommand{\topy}{19}
\newcommand{\topyy}{25}
\newcommand{\topyyy}{30}
\put(\rtwo,\toprow){\makebox(0,1){$~\;P\perp$}}
\put(\rperp,\toprow){\makebox(0,0){$P$}}
\put(\rtensor,\midrow){\makebox(0,0){$\tensor$}}
\put(\ess,\toprow){\makebox(0,1){$~\;P\perp$}}
\put(\spar,\botrow){\makebox(0,0){$\parr$}}
\put(\sperp,\toprow){\makebox(0,0){$P$}}
\put(49,3){\line(-3,4){5}}
\put(53,3){\line(3,4){5}}
\put(60,-12){\line(-3,4){7}}
\put(64,-12){\line(3,4){16}}
\end{picture}
\end{center}
Sequents are proved using the following rules:
\begin{center}
\vspace{1ex}
\hspace{2ex}
\begin{prooftree}\thickness=.08em
\strut
\justifies
\,P, P\perp
\using\axlabel
\end{prooftree}
\hspace{8.5ex}
\begin{prooftree}\thickness=.08em
\Gamma,A\;\;\;\;\;B,\Delta
\justifies
\Gamma,A\tensor B,\Delta
\using\tensorlabel
\end{prooftree}
\hspace{5.6ex}
\begin{prooftree}\thickness=.08em
\Gamma,\,A,\,B
\justifies
\Gamma,\,A\smallparr B
\using\smallparlabel
\vspace{-1ex}
\end{prooftree}
\hspace{1ex}

\vspace{3ex}

\begin{prooftree}\thickness=.08em
\Gamma,A\;\;\;\;\;\Gamma,B
\justifies
\Gamma,A\with B\using\withlabel
\end{prooftree}
\hspace{6ex}
\begin{prooftree}\thickness=.08em
\Gamma,\,A
\justifies
\Gamma,\,A\plus B\using\leftpluslabel
\end{prooftree}
\hspace{6ex}
\begin{prooftree}\thickness=.08em
\Gamma,\,B
\justifies
\Gamma,\,A\plus B\using\rightpluslabel
\end{prooftree}
\hspace{6ex}
\begin{prooftree}\thickness=.08em
\Gamma\;\;\;\;\;\Delta
\justifies
\Gamma,\Delta \using \mixlabel\hspace{2ex}\text{(optional)}
\end{prooftree}
\vspace{1ex}
\end{center}
The $\mixlabel$-rule is optional and absent by default.
Our treatment is valid for \cutfreeMALL with and without $\mixlabel$.

Throughout this document $P,Q,R$ range over
propositional variables, $A,B,\ldots$ over formulas, and
$\Gamma\mkern-3mu,\Delta,\Sigma$ over sequents.
Each of the proof rules above yields an 
implicit tracking of subformula occurrences, mapping the vertices in
the hypotheses to the ones in the conclusion.
A formula occurrence in the conclusion of a rule $\rho$ is \defn{generated}
by $\rho$ if it is not in the image of this map.

\section{Function from proofs to proof nets}\label{sec-proofnets}

A \defn{link} on $\Gamma$ is a pair (two-element set) of leaves in $\Gamma$.
A \defn{linking} on $\Gamma$ is a set of links on
$\Gamma$.\footnote{The paper \cite{HvG} imposed additional conditions
in the definition of a linking.  We do not need these conditions
here.}
Every \cutfreeMALL proof $\Pi$ of $\Gamma$ defines a set $\theta_\Pi$
of linkings on $\Gamma$ as follows.
Define a \defn{$\with$-resolution} $R$ of $\Pi$ to be any result of
deleting one branch above each $\with$-rule of $\Pi$.
By downwards tracking of formula leaves,
the axiom rules of $R$ determine a linking $\lambda_R$ on $\Gamma$.
Define $\theta_\Pi=\{\lambda_R:R\text{ is a }\with\text{-resolution of
}\Pi\}$.

Table~\ref{cut-free-induction} defines the same function by induction.
See Figure~\ref{intro-translation-eg} for an example.  
The fact that this yields the same linking set as the resolution-based
function follows from a simple structural induction on proofs.
Note that $\tensor$ (resp.\ $\with$) is \defn{multiplicative} (resp.\
\defn{additive}): multiply (resp.\ add) the number of linkings in
$\theta$ and $\theta'$ to obtain the number of linkings on the
conclusion.\footnote{This observation relies on $\theta$ and $\theta'$
having no common linking, which follows (by structural induction) from
the fact that in any proof net on $\Gamma$, every linking touches
every formula in $\Gamma$ (\ie, for every linking $\lambda$ in the
proof net, and every formula(-occurrence) $A$ in $\Gamma$, some link
of $\lambda$ contains a leaf of $A$).}

A linking set is a \defn{proof net} if it is the translation of a
proof.\footnote{In \cite{HG03,HvG} we defined a proof net via a
  geometric criterion on a linking set, and proved that a linking
  set meets this criterion if and only if it is the translation of a proof.}

\begin{table}\small
\renewcommand{\tensor}{\tighttensor}
\renewcommand{\parr}{\smallparr}
\begin{displaymath}
\prooftree\thickness=0.08em
\strut\justifies
     {\big\{
     \raisebox{-.25ex}{\footnotesize$\begin{picture}(0,10)\link{4}{13}{7}{9}\end{picture}P,P\perp$}\,%
     \big\}\judgement P,P\perp}
\using \axlabel
\endprooftree
\hspace{6ex}
\prooftree
\theta\judgement\Gamma\!,\:A\;\;\;\;\;\theta'\judgement\Gamma\!,\:B
\justifies
\theta\cup\theta'\judgement\Gamma\!,\:A\with B
\using\withlabel
\endprooftree
\hspace{6ex}
\begin{prooftree}\thickness=.08em
\theta\judgement\Gamma\!,\:A\;\;\;\;\;\;\;\;\;\;\;\;\;\;\;\;\;\;\;\;\theta'\judgement B,\Delta
\justifies
\{\lambda\cup\lambda'\,:\,\lambda\in\theta,
     \lambda'\in\theta'\}\judgement \Gamma\!,\:A\tensor
     B,\Delta
\using\tensorlabel
\end{prooftree}
\end{displaymath}
\vspace{1ex}
\begin{displaymath}
\prooftree\thickness=0.08em
\theta\judgement\Gamma\!,\:A,B \justifies
\theta\judgement\Gamma\!,\:A\parr B
\using\parlabel
\endprooftree
\hspace{6ex}
\begin{prooftree}\thickness=.08em
\theta\judgement\Gamma\!,\:A
\justifies
\theta\judgement\Gamma\!,\:A\plus B
\using\leftpluslabel
\end{prooftree}
\hspace{6ex}
\begin{prooftree}\thickness=.08em
\theta\judgement\Gamma\!,\:B
\justifies
     \theta\judgement\Gamma\!,\:A\plus B
\using\rightpluslabel
\end{prooftree}
\hspace{6ex}
\begin{prooftree}\thickness=.08em
\theta\judgement\Gamma\;\;\;\;\;\;\;\;\;\;\;\;\;\;\;\;\;\;\;\;\theta'\judgement \Delta
\justifies
\{\lambda\cup\lambda'\,:\,\lambda\in\theta,
     \lambda'\in\theta'\}\judgement \Gamma\!,\Delta
\using\mixlabel
\end{prooftree}
\end{displaymath}
\caption{\label{cut-free-induction}%
Alternative but equivalent definition of the function from
\cutoption{cut-free }\cutfreeMALL proofs to linking sets.  Here
$\theta\protect\judgement\Gamma$ signifies that $\theta$ is a set of
linkings on $\Gamma$.  We use the implicit tracking of formula leaves
downwards through rules.  The base case is a singleton linking set
whose only linking comprises a single link, between $P$ and $P\perp$.}
\figurerule\end{table}

\section{Rule commutations}\label{sec-commutations}

Tables~\ref{diag-commutations}, \ref{nondiag-commutations} and
\ref{mix-commutations}
exhaustively list the rule commutations of \cutoption{cut-free }\cutfreeMALL.  Each
commutation may be applied in context, \ie, to any subproof.
This collection of rule commutations is not ad hoc: they are
generated systematically from a general
definition of commutation, presented in
the Appendix, which is more
liberal than the one analysed by Kleene \cite{Kle52} and Curry
\cite{Cur52} in the context of sequent calculus \cite[Def.\
5.2.1]{TS96}.
\begin{table}
\small
\vspace{-2ex}
\begin{displaymath}
\renewcommand{\parr}{\footnoteparr}
\renewcommand{\convHgap}{\hspace{4ex}}
\begin{array}{c@{\convHgap}c@{\convHgap}c}
\begin{prooftree}\thickness=.08em
\[ \piproof{}{\Gamma,\A,\B,\X,\Y}
   \justifies \Gamma,\A\parr \B,\X,\Y\using\parlabel \]
\justifies \Gamma,\A\parr \B,\X\parr \Y\using\parlabel
\end{prooftree}
&\diagconv{\convparr}&
\begin{prooftree}\thickness=.08em
\[ \piproof{}{\Gamma,\A,\B,\X,\Y}
   \justifies \Gamma,\A,\B,\X\parr \Y\using\parlabel \]
\justifies \Gamma,\A\parr \B,\X\parr \Y\using\parlabel
\end{prooftree}
\\\\[\convVgap]
\begin{prooftree}\thickness=.08em
\[ \piproof{}{\Gamma,A_i,B_j}
   \justifies \Gamma,A_1\plus A_2,B_j\using\ipluslabel \]
\justifies \Gamma,A_1\plus A_2,B_1\plus B_2\using\jpluslabel
\end{prooftree}
&\diagconv{\plus}&
\begin{prooftree}\thickness=.08em
\[ \piproof{}{\Gamma,A_i,B_j}
   \justifies \Gamma,A_i,B_1\plus B_2\using\jpluslabel \]
\justifies \Gamma,A_1\plus A_2,B_1\plus B_2\using\ipluslabel
\end{prooftree}
\\\\[\convVgap]
\begin{prooftree}\thickness=.08em
\piproof{1}{\Gamma,\A}
\[
 \piproof{2}{\B,\Delta,\X} \piproof{3}{\Y,\Sigma}
 \justifies \B,\Delta,\X\tensor \Y,\Sigma \using \tensorlabel
\]
\justifies
\Gamma,\A\tensor \B,\Delta,\X\tensor \Y,\Sigma
\using \tensorlabel
\end{prooftree}
&\diagconv{\tensor}&
\begin{prooftree}\thickness=.08em
\[
 \piproof{1}{\Gamma,\A} \piproof{2}{\B,\Delta,\X}
 \justifies \Gamma,\A\tensor \B,\Delta,\X \using \tensorlabel
\]
\piproof{3}{\Y,\Sigma}
\justifies
\Gamma,\A\tensor \B,\Delta,\X\tensor \Y,\Sigma
\using \tensorlabel
\end{prooftree}
\\\\[\convVgap]
\hspace{-1ex}
\begin{prooftree}\thickness=.08em
\[
 \piproof{1}{\Gamma,\A,\X} \piproof{2}{\Gamma,\B,\X} \justifies
 \Gamma,\A\with \B,\X \using \withlabel
\]
\[
 \piproof{3}{\Gamma,\A,\Y} \piproof{4}{\Gamma,\B,\Y} \justifies
 \Gamma,\A\with \B,\Y \using \withlabel
\]
\justifies
\Gamma,\A\with \B,\X\with \Y
\using \withlabel
\end{prooftree}
&\diagconv{\with}&
\begin{prooftree}\thickness=.08em
\[
 \piproof{1}{\Gamma,\A,\X} \piproof{3}{\Gamma,\A,\Y} \justifies
 \Gamma,\A,\X\with \Y \using \withlabel
\]
\[
 \piproof{2}{\Gamma,\B,\X} \piproof{4}{\Gamma,\B,\Y} \justifies
 \Gamma,\B,\X\with \Y \using \withlabel
\]
\justifies
\Gamma,\A\with \B,\X\with \Y
\using \withlabel
\end{prooftree}
\\\\[\convVgap]
\end{array}
\end{displaymath}
\caption{\label{diag-commutations}Homogeneous rule commutations.  In the last conversion,
note the reversal of $\Pi_2$ and $\Pi_3$.}\figurerule\end{table}
\begin{table}\small\begin{displaymath}
\renewcommand{\parr}{\footnoteparr}
\renewcommand{\convVgap}{2.9ex}
\begin{array}{c@{\convHgap}c@{\convHgap}c}
\begin{prooftree}\thickness=.08em
\[ \piproof{}{\Gamma,A_i,\X,\Y}
   \justifies \Gamma,A_1\plus A_2,\X,\Y\using\ipluslabel \]
\justifies \Gamma,A_1\plus A_2,\X\parr \Y\using\parlabel
\end{prooftree}
&\conv{\plus}{\convparr}&
\begin{prooftree}\thickness=.08em
\[ \piproof{}{\Gamma,A_i,\X,\Y}
   \justifies \Gamma,A_i,\X\parr \Y\using\parlabel \]
\justifies \Gamma,A_1\plus A_2,\X\plus \Y\using\ipluslabel
\end{prooftree}
\\\\[\convVgap]
\begin{prooftree}\thickness=.08em
\[ \piproof{1}{\Gamma,A_i,\X}
   \justifies \Gamma,A_1\plus A_2,\X\using\ipluslabel \]
\[ \piproof{2}{\Gamma,A_i,\Y}
   \justifies \Gamma,A_1\plus A_2,\Y\using\ipluslabel \]
\justifies \Gamma,A_1\plus A_2,\X\with \Y\using\withlabel
\end{prooftree}
&\conv{\plus}{\with}&
\begin{prooftree}\thickness=.08em
\[ \piproof{1}{\Gamma,A_i,\X}
   \piproof{2}{\Gamma,A_i,\Y}
   \justifies \Gamma,A_i,\X\with \Y\using\withlabel \]
\justifies \Gamma,A_1\plus A_2,\X\with \Y\using\ipluslabel
\end{prooftree}
\\\\[\convVgap]
\begin{prooftree}\thickness=.08em
\[ \piproof{1}{\Gamma,\A,\B,\X}
   \justifies \Gamma,\A\parr \B,\X\using\parlabel \]
\[ \piproof{2}{\Gamma,\A,\B,\Y}
   \justifies \Gamma,\A\parr \B,\Y\using\parlabel \]
\justifies \Gamma,\A\parr \B,\X\with \Y\using\withlabel
\end{prooftree}
&\conv{\convparr}{\with}&
\begin{prooftree}\thickness=.08em
\[ \piproof{1}{\Gamma,\A,\B,\X}
   \piproof{2}{\Gamma,\A,\B,\Y}
   \justifies \Gamma,\A,\B,\X\with \Y\using\withlabel \]
\justifies \Gamma,\A\parr \B,\X\with \Y\using\parlabel
\end{prooftree}
\\\\[\convVgap]
\begin{prooftree}\thickness=.08em
\piproof{1}{\Gamma,\A}
\[ \piproof{2}{\B,\Delta,B_i}\justifies \B,\Delta,B_1\plus B_2\using\ipluslabel \]
\justifies \Gamma,\A\tensor \B,\Delta,B_1\plus B_2 \using\tensorlabel
\end{prooftree}
&\conv{\plus}{\tensor}&
\begin{prooftree}\thickness=.08em
\[\piproof{1}{\Gamma,\A}\piproof{2}{\B,\Delta,B_i}\justifies
  \Gamma,\A\tensor \B,\Delta,B_i\using\tensorlabel \]
\justifies   \Gamma,\A\tensor \B,\Delta,B_1\plus B_2\using\ipluslabel
\end{prooftree}
\\\\[\convVgap]
\begin{prooftree}\thickness=.08em
\piproof{1}{\Gamma,\A}
\[  \piproof{2}{\B,\Delta,\X,\Y}
    \justifies \B,\Delta,\X\parr \Y\using\parlabel \]
\justifies \Gamma,\A\tensor \B,\Delta,\X\parr \Y\using\tensorlabel
\end{prooftree}
&\conv{\convparr}{\tensor}&
\begin{prooftree}\thickness=.08em
\[ \piproof{1}{\Gamma,\A}
   \piproof{2}{\B,\Delta,\X,\Y}
   \justifies \Gamma,\A\tensor \B,\Delta,\X,\Y\using\tensorlabel \]
\justifies \Gamma,\A\tensor \B,\Delta,\X\parr \Y\using\parlabel
\end{prooftree}
\\\\[\convVgap]
\begin{prooftree}\thickness=.08em
\piproof{1}{\Gamma,\A}
\[
 \piproof{2}{\B,\Delta,\X} \piproof{3}{\B,\Delta,\Y} \justifies
 \B,\Delta,\X\with \Y \using \withlabel
\]
\justifies
\Gamma,\A\tensor \B,\Delta,\X\with \Y
\using \tensorlabel
\end{prooftree}
&\conv{\with}{\tensor}&
\begin{prooftree}\thickness=.08em
\[
 \piproof{1}{\Gamma,\A} \piproof{2}{\B,\Delta,\X} \justifies
 \Gamma,\A\tensor \B,\Delta,\X \using \tensorlabel
\]
\[
 \piproof{1}{\Gamma,\A} \piproof{3}{\B,\Delta,\Y} \justifies
 \Gamma,\A\tensor \B,\Delta,\Y \using \tensorlabel
\]
\justifies
\Gamma,\A\tensor \B,\Delta,\X\with \Y
\using \withlabel
\end{prooftree}
\\[\convVgap]
\end{array}
\end{displaymath}
\caption{\label{nondiag-commutations}Heterogeneous rule commutations.
  The last three
  have symmetric variants, obtained by switching $A_2\tensor A_1$
  for $A_1\tensor A_2$ and exchanging hypotheses of rules from left to
  right, correspondingly. (The hypotheses are not ordered; however, we
  apply the convention that a hypothesis that contributes to one side of
  a $\tensorlabel$ or $\withlabel$ connective is drawn on that side.)
  Note that there are two copies of the subproof $\Pi_1$ on the right
  side of the final conversion.}
\figurerule
\end{table}

\begin{table}
\small
\renewcommand{\parr}{\footnoteparr}
\begin{displaymath}
\begin{array}{c@{\convHgap}c@{\convHgap}c}
\begin{prooftree}\thickness=.08em
\piproof{1}{\Gamma}
\[
 \piproof{2}{\Delta} \piproof{3}{\Sigma}
 \justifies \Delta,\Sigma \using \mixlabel
\]
\justifies
\Gamma,\Delta,\Sigma
\using \mixlabel
\end{prooftree}
&\diagconv{\mixlabel}&
\begin{prooftree}\thickness=.08em
\[
 \piproof{1}{\Gamma} \piproof{2}{\Delta}
 \justifies \Gamma,\Delta \using \mixlabel
\]
\piproof{3}{\Sigma}
\justifies
\Gamma,\Delta,\Sigma
\using \mixlabel
\end{prooftree}
\\\\[\convVgap]
\begin{prooftree}\thickness=.08em
\piproof{1}{\Gamma}
\[
 \piproof{2}{\Delta,\X} \piproof{3}{\Y,\Sigma}
 \justifies \Delta,\X\tensor \Y,\Sigma \using \tensorlabel
\]
\justifies
\Gamma,\Delta,\X\tensor \Y,\Sigma
\using \mixlabel
\end{prooftree}
&\conv{\tensor}{\mixlabel}&
\begin{prooftree}\thickness=.08em
\[
 \piproof{1}{\Gamma} \piproof{2}{\Delta,\X}
 \justifies \Gamma,\Delta,\X \using \mixlabel
\]
\piproof{3}{\Y,\Sigma}
\justifies
\Gamma,\Delta,\X\tensor \Y,\Sigma
\using \tensorlabel
\end{prooftree}
\\\\[\convVgap]
\begin{prooftree}\thickness=.08em
\piproof{1}{\Gamma}
\[ \piproof{2}{\Delta,B_i}\justifies \Delta,B_1\plus B_2\using\ipluslabel \]
\justifies \Gamma,\Delta,B_1\plus B_2 \using\mixlabel
\end{prooftree}
&\conv{\plus}{\mixlabel}&
\begin{prooftree}\thickness=.08em
\[\piproof{1}{\Gamma}\piproof{2}{\Delta,B_i}\justifies
  \Gamma,\Delta,B_i\using\mixlabel \]
\justifies
\Gamma,\Delta,B_1\plus B_2 \using\mixlabel
\using\ipluslabel
\end{prooftree}
\\\\[\convVgap]
\begin{prooftree}\thickness=.08em
\piproof{1}{\Gamma}
\[  \piproof{2}{\Delta,\X,\Y}
    \justifies\Delta,\X\parr \Y\using\parlabel \]
\justifies \Gamma,\Delta,\X\parr \Y\using\mixlabel
\end{prooftree}
&\conv{\convparr}{\mixlabel}&
\begin{prooftree}\thickness=.08em
\[ \piproof{1}{\Gamma}
   \piproof{2}{\Delta,\X,\Y}
   \justifies \Gamma,\Delta,\X,\Y\using\mixlabel \]
\justifies \Gamma,\Delta,\X\parr \Y\using\parlabel
\end{prooftree}
\\\\[\convVgap]
\begin{prooftree}\thickness=.08em
\piproof{1}{\Gamma}
\[
 \piproof{2}{\Delta,\X} \piproof{3}{\Delta,\Y} \justifies
 \Delta,\X\with \Y \using \withlabel
\]
\justifies
\Gamma,\Delta,\X\with \Y
\using \mixlabel
\end{prooftree}
&\conv{\with}{\mixlabel}&
\begin{prooftree}\thickness=.08em
\[
 \piproof{1}{\Gamma} \piproof{2}{\Delta,\X} \justifies
 \Gamma,\Delta,\X \using \mixlabel
\]
\[
 \piproof{1}{\Gamma} \piproof{3}{\Delta,\Y} \justifies
 \Gamma,\Delta,\Y \using \mixlabel
\]
\justifies
\Gamma, \Delta,\X\with \Y
\using \withlabel
\end{prooftree}
\\\\[\convVgap]
\end{array}
\end{displaymath}
\caption{\label{mix-commutations}Mix rule commutations.  
The second conversion has a symmetric variant, in which, on the
right-hand side, the mix rule applies to the hypothesis contributing
to the right argument of the tensor. Since sequents are unordered,
we do not need symmetric variants obtained by exchanging
the hypotheses of the mix rule.
Our general definition of rule commutation in
the Appendix also allows a version of $\comm{\mixlabel}{\mixlabel}$
with three applications of $\mixlabel$, two
above and one below. However, this conversion can be generated
from the top conversion listed above and is therefore not listed explicitly.}
\figurerule
\end{table}

Our main result is that the kernel of the function from
\cutoption{cut-free }\cutfreeMALL proofs to proof nets coincides precisely
with equivalence modulo rule commutations:

\begin{theorem}\label{kernel}
Two \cutoption{cut-free }\cutfreeMALL proofs translate to the same proof net if and only
if they can be converted into each other by a series of rule commutations.
\end{theorem}
We will obtain this result as a special case of Proposition~\ref{MALL* kernel}.

\section{Cut}
\label{cut-sec}

Let \fullMALL be \cutfreeMALL, as defined in Section~\ref{mall-sec},
together with the rule
\begin{center}
\begin{prooftree}\thickness=.08em
\Gamma,A\;\;\;\;\;A\perp,\Delta
\justifies
\Gamma,\Delta
\using\text{\sf cut}
\end{prooftree}
\end{center}
\RvG{Table~\ref{cut-commutations} lists the rule commutations for $\truecutlabel$;
the rule commutations for \fullMALL not involving $\truecutlabel$
are exactly the same as in the cut-free case
(Tables~\ref{diag-commutations}--\ref{mix-commutations}).}

The translation of \fullMALL proofs to proof nets \cite{HvG} goes via a
technically convenient variant \mallcut{} of \fullMALL in which cuts are
retained in sequents.
Extend \fullMALL formulas to include \defn{cuts} $A\cut A\perp$ for any
cut-free \fullMALL formula $A$, where $\:\cut\:$ is the \defn{cut connective}.
By definition $\cut$ is unordered, \ie, $A\cut A\perp=A\perp\cut A$
(in contrast to \fullMALL formulas, where connectives are ordered, \eg,
$A\tensor B\neq B\tensor A$ when $A\neq B$).
Note that $\cut$ can only occur in outermost position.

As before, a \defn{sequent} is a disjoint union of formulas (but now a
formula may be a cut $A\cut A\perp$).
\begin{table}[t]\small
\renewcommand{\parr}{\smallparr}
\renewcommand{\tensor}{\tighttensor}
\newcommand{\gap}{\hspace{8ex}}
$$
\begin{prooftree}\thickness=.08em
\strut\justifies P,\,P\perp
\using\axlabel
\end{prooftree}
\gap
\begin{prooftree}\thickness=.08em
\Gamma,\,A,B
\justifies
\Gamma,\,A\parr B
\using\parlabel
\end{prooftree}
\gap
\begin{prooftree}\thickness=.08em\Gamma,A
      \hspace{5ex}
      B,\Delta
\justifies
     \Gamma,A\tensor B,\Delta
\using\tensorlabel
\end{prooftree}
\gap
\begin{prooftree}\thickness=.08em
\Gamma,\,A
\justifies
\Gamma,\,A\tightplus B
\using\leftpluslabel
\end{prooftree}
\gap
\begin{prooftree}\thickness=.08em\Gamma,B
\justifies     
\Gamma,A\tightplus B
\using\rightpluslabel
\end{prooftree}
$$                                                                                       
\vspace{2ex}                                                                             
$$
\begin{prooftree}\thickness=.08em
\Gamma\;\;\;\;\;\Delta
\justifies
\Gamma,\Delta \using \mixlabel\hspace{2ex}\text{(optional)}
\end{prooftree}
\hspace{6ex}
\begin{prooftree}\thickness=.08em
\Gamma,\,A \hspace{6ex} A\perp\!,\Delta
\justifies
\Gamma,\:A\cut A\perp\!,\Delta
\using\cutlabel
\end{prooftree}
\hspace{6ex}
\begin{prooftree}\thickness=.08em
\Omega_1,\Gamma,\,A
\hspace{6ex}
\Omega_2,\Gamma,\,B
\justifies      
\Omega_1,\Omega_2,\Gamma,\,A\with B
\using\withlabel\hspace{2ex}
\text{($\Omega_i$ is cut-only)}
\end{prooftree}
$$
\vspace{-1.3ex}
\caption{Rules for deriving sequents in \mallcut{}. Here $P$
  ranges over propositional variables, $A,B$ range over \fullMALL formulas,
  $\Gamma,\Delta$ range over \mallcut{} sequents, and the $\Omega_i$ range
  over cut-only sequents (disjoint unions of cuts).\label{mallcut-rules}
  Note that the $\with$-rule may superimpose one or more cuts
  from its two hypotheses (the ones contained in $\Gamma$), or may leave all
  cut pairs separate (when putting all cuts in $\Omega_i$).}
\figurerule\vspace{-3pt}\end{table}%
Sequents are derived in \mallcut{} using the rules in
Table~\ref{mallcut-rules}.
The system \mallcut{} is an extension of \cutfreeMALL.  The function
taking a \cutoption{cut-free }\cutfreeMALL proof to a set of linkings on a \cutoption{cut-free }\cutfreeMALL
sequent (defined in Section~\ref{sec-proofnets},
page~\pageref{sec-proofnets}) extends in the obvious way to a function
taking a \mallcut{} proof $\Pi$ to a set of linkings $\theta_\Pi$ on a \mallcut{} sequent:
a \defn{$\with$-resolution} $R$ of $\Pi$
is any result of deleting one branch above each $\with$-rule
of $\Pi$;
by downwards tracking of formula leaves,
the axiom rules of $\Pi$ determine a linking $\lambda_R$;
define $\theta_\Pi=\{\lambda_R:R\text{ is a }\with\text{-resolution of                     
}\Pi\}$.
Alternatively, the same function can be defined inductively, by means
of a direct extension of the cut-free case in
Table~\ref{cut-free-induction} \cite{HvG}.

A linking set on a sequent $\Gamma$ is a
\defn{proof net} if it is the translation of a \mallcut{}
proof of $\Gamma$.

Every \mallcut{} proof projects to a \fullMALL proof by deleting all cuts,
thereby turning each $\cutlabel$ rule into a standard cut rule.
Let $\theta$ be a set of linkings on a sequent.  A \fullMALL proof $\Pi$
\defn{translates} to $\theta$, or is a \defn{sequentialisation} of
$\theta$, denoted $\Pi\seqrel\theta$, if $\Pi$ is the projection of a
\mallcut{} proof translating to $\theta$.

Restricted to the cut-free case, the sequentialisation relation
$\seqrel$ is a function taking a proof to a proof net, exactly the
cut-free translation defined in Section~\ref{sec-proofnets}.
In the presence of cuts, more than one proof net may correspond to the
same \fullMALL proof. Examples can be found in \cite{HvG}.

\newcommand{\cutinconv}{\mbox{\small$\mkern.5mu\cut$}}

{\begin{table}\vspace{-2ex}\small
\begin{displaymath}
\renewcommand{\parr}{\footnoteparr}
\renewcommand{\cutinconv}{\truecutlabel}
\renewcommand{\convVgap}{3.9ex}
\begin{array}{c@{\convHgap}c@{\convHgap}c}
\begin{prooftree}\thickness=.08em
\piproof{1}{\Gamma,A}
\[
 \piproof{2}{A\perp,\Delta,B} \piproof{3}{B\perp,\Sigma}
 \justifies A\perp,\Delta,\Sigma \using \truecutlabel
\]
\justifies
\Gamma,\Delta,\Sigma
\using \truecutlabel
\end{prooftree}
&\diagconv{\cutinconv}&
\begin{prooftree}\thickness=.08em
\[
 \piproof{1}{\Gamma,A} \piproof{2}{A\perp,\Delta,B}
 \justifies \Gamma,\Delta,B \using \truecutlabel
\]
\piproof{3}{B\perp,\Sigma}
\justifies
\Gamma,\Delta,\Sigma
\using \truecutlabel
\end{prooftree}
\\\\[\convVgap]
\begin{prooftree}\thickness=.08em
  \piproof{1}{\Gamma,\,A}
\[
 \piproof{2}{A\perp,\Delta,\X} \piproof{3}{\Y,\Sigma}
 \justifies A\perp,\Delta,\X\tensor \Y,\Sigma \using \tensorlabel
\]
\justifies
\Gamma,\,\Delta,\X\tensor \Y,\Sigma
\using \truecutlabel
\end{prooftree}
&\conv{\tensor}{\cutinconv} &
\begin{prooftree}\thickness=.08em
\[
 \piproof{1}{\Gamma,A} \piproof{2}{A\perp,\Delta,\X}
 \justifies \Gamma,\Delta,\X \using \truecutlabel
\]
\piproof{3}{\Y,\Sigma}
\justifies
\Gamma,\Delta,\X\tensor \Y,\Sigma
\using \tensorlabel
\end{prooftree}
\\\\[\convVgap]
\begin{prooftree}\thickness=.08em
\piproof{1}{\Gamma,A}
\[ \piproof{2}{A\perp,\Delta,B_i}\justifies A\perp,\Delta,B_1\plus B_2\using\ipluslabel \]
\justifies \Gamma,\Delta,B_1\plus B_2 \using\truecutlabel
\end{prooftree}
&\conv{\plus}{\cutinconv}&
\begin{prooftree}\thickness=.08em
\[\piproof{1}{\Gamma,A}\piproof{2}{A\perp,\Delta,B_i}\justifies
  \Gamma,\Delta,B_i\using\truecutlabel \]
\justifies   \Gamma,\Delta,B_1\plus B_2\using\ipluslabel
\end{prooftree}
\\\\[\convVgap]
\begin{prooftree}\thickness=.08em
\piproof{1}{\Gamma,A}
\[  \piproof{2}{A\perp,\Delta,\X,\Y}
    \justifies A\perp,\Delta,\X\parr \Y\using\parlabel \]
\justifies \Gamma,\Delta,\X\parr \Y\using\truecutlabel
\end{prooftree}
&\conv{\convparr}{\cutinconv}&
\begin{prooftree}\thickness=.08em
\[ \piproof{1}{\Gamma,A}
   \piproof{2}{A\perp,\Delta,\X,\Y}
   \justifies \Gamma,\Delta,\X,\Y\using\truecutlabel \]
\justifies \Gamma,\Delta,\X\parr \Y\using\parlabel
\end{prooftree}
\\\\[\convVgap]
\begin{prooftree}\thickness=.08em
\piproof{1}{\Gamma,A}
\hspace{2ex}
\[
 \piproof{2}{A\perp,\Delta,\X} \piproof{3}{A\perp,\Delta,\Y} \justifies
 A\perp,\Delta,\X\with \Y \using \withlabel
\]
\justifies
\Gamma,\Delta,\X\with \Y
\using \truecutlabel
\end{prooftree}\hspace*{-1ex}
&\conv{\with}{\truecutlabel}&
\hspace*{-2ex}\begin{prooftree}\thickness=.08em
\[
 \piproof{1}{\Gamma,A} \piproof{2}{A\perp,\Delta,\X} \justifies
 \Gamma,\Delta,\X \using \truecutlabel
\]
\hspace{1ex}
\[
 \piproof{1}{\Gamma,A} \piproof{3}{A\perp,\Delta,\Y} \justifies
 \Gamma,\Delta,\Y \using \truecutlabel
\]
\justifies
\Gamma,\Delta,\X\with \Y
\using \withlabel
\end{prooftree}\hspace*{-2ex}
\\\\[\convVgap]
\begin{prooftree}\thickness=.08em
\piproof{1}{\Gamma}
\hspace{2ex}
\[
 \piproof{2}{\Delta,B} 
 \hspace{2ex}
 \piproof{3}{B\perp,\Sigma}
 \justifies \Delta,\Sigma \using \truecutlabel
\]
\justifies
\Gamma,\Delta,\Sigma
\using \mixlabel
\end{prooftree}
&\conv{\cutinconv}{\mixlabel}&
\begin{prooftree}\thickness=.08em
\[
 \piproof{1}{\Gamma} 
 \hspace{2ex}
 \piproof{2}{\Delta,B}
 \justifies \Gamma,\Delta,B \using \mixlabel
\]
\hspace{2ex}
\piproof{3}{B\perp,\Sigma}
\justifies
\Gamma,\Delta,\Sigma
\using \truecutlabel
\end{prooftree}
\\[\convVgap]
\end{array}
\end{displaymath}
\caption{\label{cut-commutations}\fullMALL rule commutations
  involving cut.  The second conversion has a symmetric variant
  in whose right-hand side the $\truecutlabel$ rule applies to the
  hypothesis contributing to the right argument of the tensor.}
  \figurerule\end{table}}%

Let \defn{proof-net equivalence} be the smallest equivalence relation
on \fullMALL proofs such that proofs that have a common translation are equivalent.
Then our main result (Theorem~\ref{kernel}) extends to \fullMALL as follows.
\vspace{-2pt}

\begin{theorem}\label{kernel cut}
Two \fullMALL proofs are proof-net equivalent if and only
if they can be converted into each other by a series of rule commutations.
\vspace{-2pt}
\end{theorem}
The proof will be the subject of the following section\RvG{s}.

{\begin{table}\vspace{-2ex}\small
\begin{displaymath}
\renewcommand{\parr}{\footnoteparr}
\renewcommand{\convVgap}{3.9ex}
\begin{array}{c@{\convHgap}c@{\convHgap}c}
\begin{prooftree}\thickness=.08em
\piproof{1}{\Gamma,A}
\[
 \piproof{2}{A\perp,\Delta,B} \piproof{3}{B\perp,\Sigma}
 \justifies A\perp,\Delta,B\cut B\perp,\Sigma \using \cutlabel
\]
\justifies
\Gamma,A\cut A\perp,\Delta,B\cut B\perp,\Sigma
\using \cutlabel
\end{prooftree}
&\diagconv{\cutinconv}&
\begin{prooftree}\thickness=.08em
\[
 \piproof{1}{\Gamma,A} \piproof{2}{A\perp,\Delta,B}
 \justifies \Gamma,A\cut A\perp,\Delta,B \using \cutlabel
\]
\piproof{3}{B\perp,\Sigma}
\justifies
\Gamma,A\cut A\perp,\Delta,B\cut B\perp,\Sigma
\using \cutlabel
\end{prooftree}
\\\\[\convVgap]
\begin{prooftree}\thickness=.08em
  \piproof{1}{\Gamma,\,A}
\[
 \piproof{2}{A\perp,\Delta,\X} \piproof{3}{\Y,\Sigma}
 \justifies A\perp,\Delta,\X\tensor \Y,\Sigma \using \tensorlabel
\]
\justifies
\Gamma,\,A\cut A\perp,\Delta,\X\tensor \Y,\Sigma
\using \cutlabel
\end{prooftree}
&\conv{\tensor}{\cutinconv} &
\begin{prooftree}\thickness=.08em
\[
 \piproof{1}{\Gamma,A} \piproof{2}{A\perp,\Delta,\X}
 \justifies \Gamma,A\cut A\perp,\Delta,\X \using \cutlabel
\]
\piproof{3}{\Y,\Sigma}
\justifies
\Gamma,A\cut A\perp,\Delta,\X\tensor \Y,\Sigma
\using \tensorlabel
\end{prooftree}
\\\\[\convVgap]
\begin{prooftree}\thickness=.08em
\piproof{1}{\Gamma,A}
\[ \piproof{2}{A\perp,\Delta,B_i}\justifies A\perp,\Delta,B_1\plus B_2\using\ipluslabel \]
\justifies \Gamma,A\cut A\perp,\Delta,B_1\plus B_2 \using\cutlabel
\end{prooftree}
&\conv{\plus}{\cutinconv}&
\begin{prooftree}\thickness=.08em
\[\piproof{1}{\Gamma,A}\piproof{2}{A\perp,\Delta,B_i}\justifies
  \Gamma,A\cut A\perp,\Delta,B_i\using\cutlabel \]
\justifies   \Gamma,A\cut A\perp,\Delta,B_1\plus B_2\using\ipluslabel
\end{prooftree}
\\\\[\convVgap]
\begin{prooftree}\thickness=.08em
\piproof{1}{\Gamma,A}
\[  \piproof{2}{A\perp,\Delta,\X,\Y}
    \justifies A\perp,\Delta,\X\parr \Y\using\parlabel \]
\justifies \Gamma,A\cut A\perp,\Delta,\X\parr \Y\using\cutlabel
\end{prooftree}
&\conv{\convparr}{\cutinconv}&
\begin{prooftree}\thickness=.08em
\[ \piproof{1}{\Gamma,A}
   \piproof{2}{A\perp,\Delta,\X,\Y}
   \justifies \Gamma,A\cut A\perp,\Delta,\X,\Y\using\cutlabel \]
\justifies \Gamma,A\cut A\perp,\Delta,\X\parr \Y\using\parlabel
\end{prooftree}
\\\\[\convVgap]
\begin{prooftree}\thickness=.08em
\piproof{1}{\Gamma,A}
\hspace{2ex}
\[
 \piproof{2}{A\perp,\Omega_1,\Delta,\X} \piproof{3}{A\perp,\Omega_2,\Delta,\Y} \justifies
 A\perp,\Omega_1,\Omega_2,\Delta,\X\with \Y \using \withlabel
\]
\justifies
\Gamma,A\cut A\perp,\Omega_1,\Omega_2,\Delta,\X\with \Y
\using \cutlabel
\end{prooftree}\hspace*{-1ex}
&\conv{\with}{\cutinconv}&
\hspace*{-2ex}\begin{prooftree}\thickness=.08em
\[
 \piproof{1}{\Gamma,A} \piproof{2}{A\perp,\Omega_1,\Delta,\X} \justifies
 \Gamma,A\cut A\perp,\Omega_1,\Delta,\X \using \cutlabel
\]
\hspace{1ex}
\[
 \piproof{1}{\Gamma,A} \piproof{3}{A\perp,\Omega_2,\Delta,\Y} \justifies
 \Gamma,A\cut A\perp,\Omega_2,\Delta,\Y \using \cutlabel
\]
\justifies
\Gamma,A\cut A\perp,\Omega_1,\Omega_2,\Delta,\X\with \Y
\using \withlabel
\end{prooftree}\hspace*{-2ex}
\\\\[\convVgap]
\begin{prooftree}\thickness=.08em
\piproof{1}{\Gamma}
\hspace{2ex}
\[
 \piproof{2}{\Delta,B} 
 \hspace{2ex}
 \piproof{3}{B\perp,\Sigma}
 \justifies \Delta,B\cut B\perp,\Sigma \using \cutlabel
\]
\justifies
\Gamma,\Delta,B\cut B\perp,\Sigma
\using \mixlabel
\end{prooftree}
&\conv{\cutinconv}{\mixlabel}&
\begin{prooftree}\thickness=.08em
\[
 \piproof{1}{\Gamma} 
 \hspace{2ex}
 \piproof{2}{\Delta,B}
 \justifies \Gamma,\Delta,B \using \mixlabel
\]
\hspace{2ex}
\piproof{3}{B\perp,\Sigma}
\justifies
\Gamma,\Delta,B\cut B\perp,\Sigma
\using \cutlabel
\end{prooftree}
\\[\convVgap]
\end{array}
\end{displaymath}
\caption{\label{cut-star-commutations}\mallcut{} rule commutations
  involving cut.  The second conversion also has a symmetric variant
  in whose right-hand side the $\cutlabel$ cut rule applies to the
  hypothesis contributing to the right argument of the tensor. Since
  the arguments of $\cutlabel$ are unordered, we do not need symmetric
  variants obtained by exchanging the hypotheses of the $\cutlabel$
  cut rule.}  \figurerule\end{table}}%

\section[MALL* rule commutations]{\mallcut{} rule commutations}\label{star commutations}

\RvG{We will obtain Theorem~\ref{kernel cut} from a similar theorem
  for \mallcut. To this end, we need to collect the rule commutations
  for {\mallcut}.}
As in the cut-free case, \RvG{these}
can be generated systematically from the general definition of rule
commutation in the Appendix.
Table~\ref{cut-star-commutations} lists the rule commutations for $\cutlabel$.
The rule commutations for {\mallcut} not involving $\cutlabel$ or $\withlabel$
are exactly the same as in the cut-free case (Tables~\ref{diag-commutations}--\ref{mix-commutations}),
whereas the heterogeneous commutations involving $\withlabel$ are
obtained from the cut-free ones by the addition of $\Omega_1$ and
$\Omega_2$, just as in the rule $\comm{\with}{\cutinconv}$.
The commutation $\comm{\with}{\with}$ is a bit more involved:
it can be obtained from the one in Table~\ref{diag-commutations}
by the addition of \plat{$2^{2^2}-2=14$} variables $\Omega_{ijkl}$ with $i,j,k,l\in\{0,1\}$.
Each cut occurring in the conclusion can be produced by one or more of
the subproofs $\Pi_1$--$\Pi_4$. The variable $\Omega_{1001}$ captures
those cuts that are produced by $\Pi_1$ as well as $\Pi_4$, but not
by $\Pi_2$ or $\Pi_3$; in general the $n^{\rm th}$ index from the series
$i,j,k,l$ indicates whether or not the cuts in $\Omega_{ijkl}$ are
produced by $\Pi_n$. A sequent occurring in the rule commutation is
enriched with $\Omega_{ijkl}$ iff it occurs under $\Pi_n$ for an $n$
such that the \plat{$n^{\rm th}$} index from the series $i,j,k,l$ is set to 1.
The variable $\Omega_{0000}$ is not needed, as it would not occur in
the conclusion, and the variable $\Omega_{1111}$ is superfluous, as it
can be incorporated in $\Gamma$.
Since the resulting rules do not fit on the page, below they are
displayed using an abbreviation: $\Omega^n$ denotes the disjoint union
of the sequents $\Omega_{ijkl}$ where the $n^{\rm th}$ index is set to 1.
Likewise $\Omega^{mn}$ indicates the disjoint union of the sequents
$\Omega_{ijkl}$ where either the \plat{$m^{\rm th}$ or the $n^{\rm th}$}
index is set to one, i.e.\ the non-disjoint union of $\Omega^m$ and $\Omega^n$.

{\small
\[
\begin{array}{c}
\begin{prooftree}\thickness=.08em
\[
 \piproof{1}{\Omega^{1},\Gamma,\A,\X} \piproof{2}{\Omega^{2},\Gamma,\B,\X} \justifies
 \Omega^{12},\Gamma,\A\with \B,\X \using \withlabel
\]
\[
 \piproof{3}{\Omega^{3},\Gamma,\A,\Y} \piproof{4}{\Omega^{4},\Gamma,\B,\Y} \justifies
 \Omega^{34},\Gamma,\A\with \B,\Y \using \withlabel
\]
\justifies
\Omega^{1234},\Gamma,\A\with \B,\X\with \Y
\using \withlabel
\end{prooftree}
\\~
\\
\makebox[0pt][l]{{\large$\updownarrow$}~~$\comm{\with}{\with}$}
\\~
\\
\begin{prooftree}\thickness=.08em
\[
 \piproof{1}{\Omega^{1},\Gamma,\A,\X} \piproof{3}{\Omega^{3},\Gamma,\A,\Y} \justifies
 \Omega^{13},\Gamma,\A,\X\with \Y \using \withlabel
\]
\[
 \piproof{2}{\Omega^{2},\Gamma,\B,\X} \piproof{4}{\Omega^{4},\Gamma,\B,\Y} \justifies
 \Omega^{24},\Gamma,\B,\X\with \Y \using \withlabel
\]
\justifies
\Omega^{1234},\Gamma,\A\with \B,\X\with \Y
\using \withlabel
\end{prooftree}
\vspace{5pt}
\end{array}
\]
}%
The rule commutations for {\fullMALL}
\RvG{(cf.~Section~\ref{cut-sec})} are obtained from the ones of
{\mallcut} by omitting all cuts from sequents.

\section[Proof of the MALL*      rule commutation theorem]
        {Proof of the {\mallcut} rule commutation theorem}

We say that a $\beta$-rule \defn{commutes over} an $\alpha$-rule if
there is a valid {\mallcut} rule commutation where a proof fragment in which
the $\beta$-rule occurs immediately below one or more $\alpha$-rules is
replaced by a proof fragment in which this order is reversed.
Using either the definition of rule commutation from
the Appendix or the enumeration of
Tables~\ref{diag-commutations}, \ref{nondiag-commutations},
\ref{mix-commutations} and~\ref{cut-star-commutations}, enriched with
$\Omega$s as discussed above, it is not hard to check that
this happens if and only if (i) no formula occurrence generated by one
of the $\alpha$-rules tracks to a subformula of a formula generated by
the $\beta$-rule, and (ii) one of the
following cases applies (\cf\ Table~\ref{commutation-table}):
\begin{table}[tb]
\renewcommand{\parr}{\footnoteparr}
$$\begin{array}{l|ccccccc}
\raisebox{-5pt}{$\beta\!$}\raisebox{-2pt}{$\backslash$}\!\alpha
         & \cutlabel & \mixlabel & \tensor   & \plus_1   & \plus_2   & \parr     & \with     \\[5pt]
\hline
\cutlabel&\checkmark &\checkmark &\checkmark &\checkmark &\checkmark &\checkmark &\checkmark \\
\mixlabel&\checkmark &\checkmark &\checkmark &\checkmark &\checkmark &\checkmark &\checkmark \\
\tensor	 &\checkmark &\checkmark &\checkmark &\checkmark &\checkmark &\checkmark &\checkmark \\
\plus_1	 &\checkmark &\checkmark &\checkmark &\checkmark &\checkmark &\checkmark &\checkmark \\
\plus_2	 &\checkmark &\checkmark &\checkmark &\checkmark &\checkmark &\checkmark &\checkmark \\
\parr	 &\circ      &\circ      &\circ      &\checkmark &\checkmark &\checkmark &\checkmark \\
\with	 &\bullet    &\bullet    &\bullet    &\circ      &\circ      &\circ      &\circ      \\
\end{array}$$
\caption{Rule commutations.\label{commutation-table}
The check marks flag pairs \protect\raisebox{2pt}{$\frac{\alpha}{\beta}$}
where a (lower) $\beta$-rule always commutes over an $\alpha$-rule. The marks $\circ$
indicate situations where $\beta$-rules commute over $\alpha$-rules only
under certain syntactic restrictions, which can be found by studying
the results of commuting $\alpha$- over $\beta$-rules.  The $\bullet$
denotes commutation under certain syntactic restrictions.}
\figurerule
\end{table}
\begin{itemize}
\vspace{-1pt}
\item $\beta\in\{\tensor,\plus_1,\plus_2,\mixlabel,\cutlabel\}$;
\vspace{-1pt}
\item $\beta=\parr$ and $\alpha\neq\tensor,\mixlabel,\cutlabel$;
\vspace{-1pt}
\item $\beta=\parr$, $\alpha=\tensor$, $\mixlabel$ or $\cutlabel$, and both
arguments of the formula generated by the $\parr$-rule
occur in the same hypothesis of the $\alpha$-rule;
\vspace{-1pt}
\item $\beta=\with$, $\alpha\neq\tensor,\mixlabel,\cutlabel$,
  and the formula occurrences generated by the two $\alpha$-rules
  track to the same same formula occurrence of the $\beta$-rule.
\vspace{-1pt}
\item $\beta=\with$, $\alpha=\tensor$, $\mixlabel$ or $\cutlabel$, the $\beta$-rule
  generates a formula $B_1\with B_2$, and the hypotheses of the two
  $\alpha$-rules that do not contain $B_1$ or $B_2$ are the same, and
  have identical subproofs.
\end{itemize}
This, in turn, yields exactly the rule commutations of
Tables~\ref{diag-commutations}--\ref{cut-star-commutations},
enriched with $\Omega$s as discussed in Section~\ref{star commutations}.

The following result, a generalisation of Theorem~\ref{kernel},
is a crucial step towards proving Theorem~\ref{kernel cut}.

\begin{proposition}\label{MALL* kernel}
Two {\mallcut} proofs translate to the same proof net if and only
if they can be converted into each other by a series of rule commutations.
\end{proposition}

\begin{proof}
If $\Pi'$ can be obtained from $\Pi$ by commuting rule occurrences,
then $\Pi$ and $\Pi'$ translate to the same linking set:
taking a $\with$-resolution on either side of a commutation
induces essentially the same $\with$-resolutions (or deletions) of the
subproofs $\Pi_i$. For example, in the last commutation in
Table~\ref{nondiag-commutations}, if we choose \emph{right} for the
distinguished $\with$-rule, we delete subproof $\Pi_2$ from both sides,
and induce corresponding $\with$-resolutions of $\Pi_1$ and $\Pi_3$.
The converse is proved below.
\end{proof}
Given a set of linkings $\Lambda$ on a {\mallcut} sequent $\Gamma$, let
$\Gamma\mathord{\restriction}\Lambda$ be obtained from the forest
$\Gamma$ by deleting all vertices that are not below
a leaf of $\Gamma$ that occurs in $\Lambda$ (\ie, in a link in a linking of $\Lambda$).
A $\with$-vertex $w$ in $\Gamma$ is \defn{toggled} by $\Lambda$ if
both arguments of $w$ occur in $\Gamma\mathord{\restriction}\Lambda$.
A link $a$ \defn{depends} on $w$ in $\Lambda$ if there
exist $\lambda,\lambda'\in \Lambda$ such that $a\in\lambda$,
$a\not\in\lambda'$, and $w$ is the only $\with$ toggled by
$\{\lambda,\lambda'\}$.
Construct the \defn{graph} $\jumpgraph{\Lambda}$ \cite{HvG} from
$\Gamma\mathord{\restriction}\Lambda$ by
adding the edges of $\bigcup_{\lambda\in\Lambda}\lambda$, as well
as all \defn{jump} edges from leaves $\ell$ and $\ell'$ to any $\with$-vertex
on which the link $\{\ell,\ell'\}$ depends in $\Lambda$.
Below we will need the following properties of a proof net $\theta$ on
a {\mallcut} sequent $\Gamma$, established in \cite{HvG}.
\begin{eqnarray}
&&\mbox{\it Any set of two linkings in $\theta$ toggles a
  $\with$-vertex of\/ $\Gamma$.}
\label{resolution}\\
&&\mbox{\it Each root vertex (formula occurrence) in
  $\Gamma$ occurs in $\jumpgraph{\theta}$.}
\label{inhabitation}\\
&&\begin{array}{@{}l@{}}\mbox{\it For every $\lambda\mathbin\in\theta$
 and each root $\with$-vertex $w$ in $\Gamma$}\\
 \mbox{\it there is a $\lambda'\mathbin\in\theta$ such that
  $w$ is the only $\with$ toggled by $\{\lambda,\lambda'\}$.}
  \end{array}
\label{P2}
\end{eqnarray}
A formula occurrence $A=A_1 \alpha A_2$ in a {\mallcut} sequent $\Gamma$
\defn{separates} a proof net $\theta$ on $\Gamma$
if (i) $\alpha\in\{\parr,\with\}$, (ii) $\alpha=\plus$ and one of the
$A_i$ does not occur in $\fullgraph$,
or (iii) $\alpha\in\{\tensor,\cutlabel\}$ and
$\jumpgraph\theta$ has no cycle through $\alpha$.
\begin{lemma}\label{lem-last-gen} If the last rule of a {\mallcut} proof
  generates $A$, then $A$ separates the associated proof net.
\end{lemma}
\begin{proof}
  The only non-trivial cases are \RvG{$A=A_1\alpha A_2$
  for $\alpha\in\{\tensor,\cutlabel\}$}.
  Let $\Gamma_1$ and
  $\Gamma_2$ be the hypotheses of the last rule $\rho$ of the proof $\Pi$, let
  $\Pi_i$ be the branch of $\Pi$ above $\rho$ proving $\Gamma_i$, let
  $\theta$ be the proof net of $\Pi$ and $\theta_i$ that
  of $\Pi_i$.  $\fullgraph$ could have a cycle
  through $\alpha$ only when in $\theta$ a link $a$ in $\Gamma_1$
  depends on a $\with$-vertex $w$ in $\Gamma_{2}$ (or vice versa).  In
  that case there exist $\lambda,\lambda'\in\theta$ such that $a\in
  \lambda$, $a\not\in\lambda'$, and $w$ is the only $\with$ toggled by
  $\{\lambda,\lambda'\}$. Hence there must be $\lambda_1,\lambda'_1$
  in $\theta_1$ and $\lambda_2,\lambda'_2$ in $\theta_2$ such that
  $a\in\lambda_1$, $a\not\in\lambda'_1$ and $w$ is the only $\with$
  toggled by $\{\lambda_1 \cup \lambda_2,\lambda'_1 \cup
  \lambda'_2\}$.  However, by (\ref{resolution}) there must be another
  $\with$-vertex of $\Lambda$ that is toggled by $\{\lambda_1 \cup
  \lambda_2,\lambda'_1 \cup \lambda'_2\}$, namely one occurring in
  $\Gamma_1$ that is toggled by $\{\lambda_1,\lambda'_1\}$.
\end{proof}
\begin{lemma}\label{lem-conn}
  If a formula occurrence $\,A=A_1 \alpha A_2\,$ in a \mallcut\/ sequent\/
  $\,\Gamma\!,\,A$ separates a proof net $\theta$ of\/ $\,\Gamma\!,\,A$ for which
  $\fullgraph$ is connected,
  then there is at most one instance $\sigma$ of an $\alpha$-rule that
  could generate $A$ in the last step of a proof\/ $\Pi$
  of\/ $\,\Gamma\!,\,A$ with proof net $\theta$.
\end{lemma}
\begin{trivlist}\item{}\normalfont\textit{Proof. }
\begin{itemize}
\item Case $\alpha = \parr$: the hypothesis of $\sigma$ must be
$\Gamma\!,\, A_1, A_2$.

\item Case $\alpha = \with$: the hypotheses of $\sigma$ must be
$\Gamma\!,\, A_1$ and $\Gamma\!,\, A_2$.

\item 
Case $\alpha = \plus$: exactly one of the $A_i$, say $A_d$, is in
$\fullgraph$ (\ref{inhabitation}).
Hence the hypothesis of $\sigma$ must be $\Gamma\!,\,A_d$.

\item 
Case $\alpha\in\{\tensor,\cutlabel\}$: let $\Gamma\!,\, A_1, A_2$ be the sequent
resulting from deleting the connective $\alpha$ in $A$ from $\Gamma\!,\, A$. Since
$A$ separates $\theta$ and $\fullgraph$ is connected,
the restriction of $\jumpgraph\theta$ to $\Gamma\!,\, A_1, A_2$ has two disconnected components, one on a sequent $\Gamma_1,\,
A_1$ and the other on a sequent $\Gamma_2,\, A_2$, where \mbox{$\Gamma_1 \cup
\Gamma_2 = \Gamma$}.  Using (\ref{inhabitation}), the
hypotheses of $\sigma$ must be $\Gamma_1,\, A_1$ and $\Gamma_2,\, A_2$.
\hfill$\square$
\end{itemize}
\end{trivlist}
In each case the proof nets on the hypotheses of $\sigma$, induced by
the branches of $\Pi$ that prove these hypotheses, are completely
determined by $\theta$.

\vspace{1ex}

\noindent
For $\Pi$ a {\mallcut} proof, let $\jumpgraph\Pi$ abbreviate $\jumpgraph{\theta_\Pi}$.
We shall prove the following four lemmas by simultaneous structural
induction.
\begin{lemma}\label{lem-path}
  Let $\Pi$ be a proof of a {\mallcut} sequent $\Delta,A_1\alpha
  A_2,\Sigma$ such that in $\jumpgraph\Pi$ any path between (vertices
  in) $\Delta,A_1$ and $A_2,\Sigma$ passes through the indicated
  occurrence of $\alpha\in\{\tensor,\cutlabel\}$.  Then $\Pi$ can, by means of rule
  commutations, be converted into a proof $\Pi''$ whose last step is
  the $\alpha$-rule with hypotheses $\Delta,A_1$ and $A_2,\Sigma$.
\end{lemma}
\begin{lemma}\label{lem-sep}
Let $\Pi$ be a proof of a {\mallcut} sequent $\Gamma$ whose proof
net $\theta$ is separated by a formula occurrence $A$ in $\Gamma$. Then, by means of a
series of rule commutations, $\Pi$ can be converted into a proof
$\Pi''$ of $\Gamma$ that generates $A$ in its last step.
\end{lemma}
\begin{lemma}\label{lem-no-path}
  Let $\Pi$ be a proof of a {\mallcut} sequent $\Delta,\Sigma$ for nonempty
  sequents $\Delta$ and $\Sigma$, such that in $\jumpgraph\Pi$ there
  is no path between (vertices in) $\Delta$ and $\Sigma$.  Then $\Pi$
  can, by means of rule commutations, be converted into a proof
  $\Pi''$ whose last step is the $\mixlabel$-rule with hypotheses
  $\Delta$ and $\Sigma$.
\end{lemma}
\begin{lemma}\label{lem-same}
If two proofs $\Pi$ and $\Pi'$ of a {\mallcut} sequent $\Gamma$
translate to the same proof net on $\Gamma$, then $\Pi$ can
be converted into $\Pi'$ by a series of rule commutations.
\end{lemma}
Lemma~\ref{lem-same} is the
\RvG{converse direction of}
Proposition~\ref{MALL* kernel}
\RvG{that} must be established.
\begin{proof}
  We prove Lemmas~\ref{lem-path}--\ref{lem-same} by a simultaneous structural
  induction on $\Pi$ (or equivalently, on $\Gamma$).
\vspace{1ex}

\noindent
\emph{Induction base (applies to Lemma~\ref{lem-same} only).}  The induction base
 is trivial, as a \fullMALL sequent that can be proven in one step
 has at most one proof, a single application of $\axlabel$.

\vspace{1ex}

\noindent
\emph{Induction step for Lemma~\ref{lem-path}.}
\begin{itemize}
\item
First consider the case that the last step $\rho$ of $\Pi$ is an
application of $\mixlabel$, say with hypotheses $\Gamma_c$ and
$\Gamma_d,A_1\RvG{\alpha} A_2$.

Let $\Pi_d$ be the branch of $\Pi$ above $\rho$ proving
$\Gamma_d,A_1\RvG{\alpha} A_2$.
Let $\Delta_d=\Delta \cap \Gamma_d$ and $\Sigma_d = \Sigma\cap\Gamma_d$.
Since $\jumpgraph{\Pi_d}$ is a subgraph of $\jumpgraph\Pi$,
any path in $\jumpgraph{\Pi_d}$ between (vertices in)
$\Delta_d,A_1$ and $A_2,\Sigma_d$ passes through the indicated occurrence
of $\RvG{\alpha}$.
Hence, by induction, 
$\Pi_d$ can, by means of rule commutations, be converted into a
proof $\Pi_d'$ whose last step is the $\RvG{\alpha}$-rule
with hypotheses $\Delta_d,A_1$ and $A_2,\Sigma_d$.

Let $\Pi_c$ be the branch of $\Pi$ above $\rho$ proving $\Gamma_c$.
Let $\Delta_c=\Delta \cap \Gamma_c$ and $\Sigma_c =
\Sigma\cap\Gamma_c$.  Since $\jumpgraph{\Pi_c}$ is a subgraph of
$\jumpgraph\Pi$, there is no path in $\jumpgraph{\Pi_c}$ between
(vertices in) $\Delta_c$ and $\Sigma_c$. If $\Delta_c$ or $\Sigma_c$
is empty, let $\Pi'_c = \Pi_c$. Otherwise, by induction, using Lemma~\ref{lem-no-path},
$\Pi_c$ can, by means of rule commutations, be converted into a
proof $\Pi_c'$ whose last step is the $\mixlabel$-rule with hypotheses
$\Delta_c$ and $\Sigma_c$.

Let $\Pi'$ be the proof obtained from $\Pi$ by replacing
$\Pi_d$ with $\Pi_d'$ and $\Pi_c$ with $\Pi_c'$. Let $\Pi''$ be the
proof with the same 3 or 4 subproofs yielding $\Delta_c$,~ $\Sigma_c$,~
$\Delta_d,A_1$ and $A_2,\Sigma_d$ that first combines
$\Delta_c$ with $\Delta_d,A_1$ into $\Delta,A_1$ using $\mixlabel$
(provided $\Delta_c$ is nonempty), and likewise combines
$\Sigma_c$ with $A_2,\Sigma_d$ into $A_2,\Sigma$ using $\mixlabel$
(provided $\Sigma_c$ is nonempty), and then applies $\RvG{\alpha}$ to
yield $\Delta,A_1\RvG{\alpha} A_2,\Sigma$. By means of a few simple rule
commutations, $\Pi'$ can be converted into $\Pi''$.
\item
Next consider the case that the last step $\rho$ of $\Pi$ is an
application of $\RvG{\alpha}$ generating the same formula $A_1\RvG{\alpha} A_2$.
Let the hypotheses of $\rho$ be $\Gamma_i, A_i$ for $i=1,2$.

Let $\Pi_i$ be the branch of $\Pi$ above $\rho$ proving
$\Gamma_i,A_i$.
Let $\Delta_i=\Delta \cap \Gamma_i$ and $\Sigma_i = \Sigma\cap\Gamma_i$.
Since $\jumpgraph{\Pi_1}$ is a subgraph of $\jumpgraph\Pi$,
there is no path in $\jumpgraph{\Pi_1}$ between (vertices in)
$\Delta_1,A_1$ and $\Sigma_1$.
In case $\Sigma_1$ is empty, let $\Pi_1'=\Pi_1$. Otherwise, by
induction, using Lemma~\ref{lem-no-path}, $\Pi_1$ can, by means of rule commutations,
be converted into a proof $\Pi_1'$ whose last step is the $\mixlabel$-rule
with hypotheses $\Delta_1,A_1$ and $\Sigma_1$.

In case $\Delta_2$ is empty, let $\Pi_2'=\Pi_2$. Otherwise, by
means of rule commutations, $\Pi_2$ can be converted into a
proof $\Pi_2'$ whose last step is the $\mixlabel$-rule
with hypotheses $\Delta_2$ and $A_2,\Sigma_2$.

Let $\Pi'$ be the proof obtained from $\Pi$ by replacing
$\Pi_i$ with $\Pi_i'$ for $i\in\{1,2\}$. Let $\Pi''$ be the
proof with the same 2, 3 or 4 subproofs yielding $\Delta_1,A_1$,~
$\Sigma_1$,~ $\Delta_2$ and $A_2,\Sigma_2$ that first combines
$\Delta_1,A_1$ with $\Delta_2$ into $\Delta,A_1$ using $\mixlabel$
(provided $\Delta_2$ is nonempty), and likewise combines
$\Sigma_1$ with $A_2,\Sigma_2$ into $A_2,\Sigma$ using $\mixlabel$
(provided $\Sigma_1$ is nonempty), and then applies $\RvG{\alpha}$ to
yield $\Delta, A_1\RvG{\alpha} A_2, \Sigma$. By means of a few simple rule
commutations, $\Pi'$ can be converted into $\Pi''$.
\end{itemize}
In the remaining cases let the last step of $\Pi$ be a
a $\beta$-rule $\rho$ generating the formula $B = B_1 \beta B_2 \neq
A = A_1\RvG{\alpha} A_2$. We treat the case that $B$ occurs in $\Sigma$;
the other case follows by symmetry. Let $\Sigma = \Sigma',B_1\beta B_2$.
\begin{itemize}
\item
Let $\beta=\plus$.
Let $\Pi_d$ be the part of $\Pi$ above $\rho$, proving the hypothesis
$\Delta, A, \Sigma', B_d$ of $\rho$ (where $d$ is $1$ or $2$).
Since $\jumpgraph{\Pi_d}$ is a subgraph of $\jumpgraph\Pi$,
any path in $\jumpgraph{\Pi_d}$ between (vertices in)
$\Delta,A_1$ and $A_2,\Sigma',B_d$ passes through the indicated occurrence
of $\RvG{\alpha}$.  Thus, by induction, by a series of
rule commutations $\Pi_d$ can be be converted into a proof $\Pi'_d$ of
$\Delta, A_1\RvG{\alpha} A_2, \Sigma', B_d$ whose last step is the $\RvG{\alpha}$-rule
with hypotheses $\Delta,A_1$ and $A_2,\Sigma',B_d$. Let $\Pi'$ be
the proof obtained from $\Pi$ by replacing $\Pi_d$ by $\Pi_d'$.
In $\Pi'$, $\rho$ commutes over the $\RvG{\alpha}$-rule generating $A$,
thereby yielding the required proof $\Pi''$.
\item
Let \RvG{$\beta\in\{\tensor,\cutlabel\}$}.
Let $\Pi_1$ and $\Pi_2$ be the branches of $\Pi$ above $\rho$ proving
the hypotheses $\Delta_1, A, \Sigma_1, B_1$ and $\Delta_2, \Sigma_2,
B_2$ of $\rho$, respectively.
Here $\Delta = \Delta_1, \Delta_2$ and $\Sigma' = \Sigma_1,\Sigma_2$.
We assume that $A$ sides with $B_1$; the other case proceeds symmetrically. 
Since $\jumpgraph{\Pi_1}$ is a subgraph of $\jumpgraph\Pi$,
any path in $\jumpgraph{\Pi_1}$ between (vertices in)
$\Delta_1,A_1$ and $A_2,\Sigma_1,B_1$ passes through the indicated occurrence
of $\RvG{\alpha}$.  Thus, by induction, by a series of
rule commutations $\Pi_1$ can be be converted into a proof $\Pi'_1$ of
$\Delta_1, A_1\RvG{\alpha} A_2, \Sigma_1, B_1$ whose last step is the $\RvG{\alpha}$-rule
with hypotheses $\Delta_1,A_1$ and $A_2,\Sigma_1,B_1$.

In case $\Delta_2$ is empty, let $\Pi_2' \mathbin{=} \Pi_2$. Otherwise, by
induction, using Lemma~\ref{lem-no-path}, $\Pi_2$ can be be converted into a proof $\Pi'_2$ of
$\Delta_2, \Sigma_2, B_2$ whose last step is the $\mixlabel$-rule
with hypotheses $\Delta_2$ and $\Sigma_2,B_2$.

Let $\Pi'$ be the proof obtained from $\Pi$ by replacing
$\Pi_i$ by $\Pi_i'$, for $i\in\{1,2\}$.
Let $\Pi''$ be the proof with the same 3 or 4 subproofs yielding
$\Delta_1,A_1$,~ $A_2,\Sigma_1,B_1$,~ $\Delta_2$ and $\Sigma_2,B_2$
that first combines $\Delta_2$ with $\Delta_1,A_1$ into $\Delta,A_1$
using $\mixlabel$ (provided $\Delta_2$ is nonempty), and likewise combines
$\Sigma_2,B_2$ with $A_2,\Sigma_1,B_1$ into $A_2,\Sigma',B$ using
\RvG{$\beta$}, and then applies $\RvG{\alpha}$ to yield $\Delta,A,\Sigma',B$.
By means of a few simple rule commutations, $\Pi'$ can be converted into $\Pi''$.
\item 
Let $\beta = \parr$.
Let $\Pi_\rho$ be the part of $\Pi$ above $\rho$. Then $\Pi_\rho$
proves the hypothesis $\Delta, A, \Sigma', B_1, B_2$ of $\rho$.
Since $\jumpgraph{\Pi_\rho}$ is a subgraph of $\jumpgraph\Pi$, in
$\jumpgraph{\Pi_\rho}$ any path between (vertices in) $\Delta,A_1$ and
$A_2,\Sigma',B_1, B_2$ passes through the indicated occurrence of $\RvG{\alpha}$.
Hence, by induction, using by Lemma~\ref{lem-path}, $\Pi_\rho$ can, by means of
rule commutations, be converted into a proof $\Pi_\rho'$ whose last step is
the $\RvG{\alpha}$-rule with hypotheses $\Delta,A_1$ and $A_2,\Sigma',B_1,B_2$. 
Let $\Pi'$ be the proof obtained from $\Pi$ by replacing $\Pi_\rho$
by $\Pi_\rho'$. In $\Pi'$ the $\parr$-rule $\rho$ commutes
over the $\RvG{\alpha}$-rule generating $A$, thereby yielding the required
proof $\Pi''\!$.
\item 
Let $\beta = \with$.
The rule $\rho$ has hypotheses
\begin{figure}[h]
\vspace{-1ex}
\begin{center}
\begin{prooftree}\thickness=.08em
  \[
      \using                          
              \Pi_1
      \proofdotseparation=1.2ex       
      \proofdotnumber=4
      \leadsto                        
          \RvG{\Omega^\Delta_1,\Delta'},A_1\RvG{\alpha} A_2,\RvG{\Omega^\Sigma_1,\Sigma''},B_1
  \]
  \hspace{4ex}
  \[
      \using                          
              \Pi_2
      \proofdotseparation=1.2ex       
      \proofdotnumber=4
      \leadsto                        
          \RvG{\Omega^\Delta_2,\Delta'},A_1\RvG{\alpha} A_2,\RvG{\Omega^\Sigma_2,\Sigma''},B_2
  \]
  \hspace{-1ex}
  \justifies \RvG{\Omega^\Delta_1,\Omega^\Delta_2,\Delta'},A_1\RvG{\alpha} A_2,
             \RvG{\Omega^\Sigma_1,\Omega^\Sigma_2,\Sigma''},B_1\with B_2\using \withlabel(\rho)
\end{prooftree}
\vspace*{-1em}
\end{center}
\end{figure}
hypotheses $\RvG{\Omega^\Delta_1,\Delta'}, A, \RvG{\Omega^\Sigma_1,\Sigma''}, B_1$ and
$\RvG{\Omega^\Delta_2,\Delta'}, A, \RvG{\Omega^\Sigma_2,\Sigma''}, B_2$
\RvG{with $\Delta = \Omega^\Delta_1,\Omega^\Delta_2,\Delta'$
and $\Sigma' = \Omega^\Sigma_1,\Omega^\Sigma_2,\Sigma''$}.
\RvG{We claim that $\Omega^\Delta_1$, and by symmetry also $\Omega^\Delta_2$, is empty. For if not, let $\ell$ be a
leaf in $\Omega^\Delta_1$ that occurs in a link $a$ in a linking $\nu$ of $\jumpgraph{\Pi_1}$---such a leaf
exists by (\ref{inhabitation}). Then, using Table~\ref{cut-free-induction}, $\nu$ also occurs in $\jumpgraph\Pi$.
Using (\ref{P2}), let $\nu' \in \theta_\Pi$ be such that $\beta$ is the only
$\with$ toggled by $\{\nu, \nu'\}$. Again using
Table~\ref{cut-free-induction}, $\nu'$ must occur in $\jumpgraph{\Pi_2}$.
Since $\ell$ does not occur in $\jumpgraph{\Pi_2}$, $a$ cannot occur
in $\nu'$, and thus depends on $\beta$. Hence in $\jumpgraph\Pi$ there is a jump
edge from $\ell$ to $\beta$.
This contradicts the assumption that in $\fullgraph$ any path between
(vertices in) $\Delta,A_1$ and $A_2,\Sigma',B_1\parr B_2$ passes through
the indicated occurrence of $\RvG{\alpha}$.}

Let $\Pi_i$ be the branch of $\Pi$ above $\rho$ proving
$\RvG{\Delta'}, A, \RvG{\Omega^\Sigma_i},\linebreak[1]\Sigma'', B_i$.
Since $\jumpgraph{\Pi_i}$ is a subgraph of $\jumpgraph\Pi$, in
$\jumpgraph{\Pi_i}$ any path between (vertices in) $\RvG{\Delta'},A_1$ and
$A_2,\RvG{\Omega^\Sigma_i,\Sigma''},B_i$ passes through the indicated occurrence of $\RvG{\alpha}$.
Hence, by induction, using Lemma~\ref{lem-path}, $\Pi_i$ can, by means of rule
commutations, be converted into a proof $\Pi_i'$ whose last step is the
$\RvG{\alpha}$-rule with hypotheses $\RvG{\Delta'},A_1$ and
$A_2,\RvG{\Omega^\Sigma_i,\Sigma''},B_i$.

Thus the left hypotheses of $\Pi'_1$ and $\Pi'_2$
are both $\RvG{\Delta'},A_1$, and we claim that the proof nets on them induced
by the subproofs $\Pi'_{11}$ and $\Pi'_{21}$ of $\Pi$ leading up to these
hypotheses must be the same.
\begin{center}
\begin{prooftree}\thickness=.08em
  \[
    \[
      \using                          
              \Pi_{11}'
      \proofdotseparation=1.2ex       
      \proofdotnumber=4
      \leadsto                        
          \RvG{\Delta'},A_1
    \]
    \hspace{4ex} 
    A_2,\RvG{\Omega^\Sigma_1,\Sigma''},B_1
    \justifies \RvG{\Delta'},A_1\RvG{\alpha} A_2,\RvG{\Omega^\Sigma_1,\Sigma''},B_1
    \using\RvG{\alpha}
  \]
  \hspace{4ex}
  \[
    \[
      \using                          
              \Pi_{21}'
      \proofdotseparation=1.2ex       
      \proofdotnumber=4
      \leadsto                        
          \RvG{\Delta'},A_1
    \]
    \hspace{4ex} 
    A_2,\RvG{\Omega^\Sigma_2,\Sigma''},B_2
    \justifies \RvG{\Delta'},A_1\RvG{\alpha} A_2,\RvG{\Omega^\Sigma_2,\Sigma''},B_2
    \using\RvG{\alpha}
  \]
  \hspace{-1ex}
  \justifies \RvG{\Delta'},A_1\RvG{\alpha} A_2,\RvG{\Omega^\Sigma_1,\Omega^\Sigma_2,\Sigma''},B_1\with B_2\using \withlabel(\rho)
\end{prooftree}
\vspace{1ex}
\end{center}
For if not, let $\lambda$ be a linking in the proof net of $\Pi'_{11}$
but not in the proof net of $\Pi'_{21}$. (The symmetric case goes
likewise.) Then, using Table~\ref{cut-free-induction}, for some linking $\mu$ on
$A_2,\RvG{\Omega^\Sigma_1,\Sigma''},B_1$, the linking $\nu := \lambda \cup \mu$ must be in
the proof net $\theta$ of $\Pi$.
Using (\ref{P2}), let $\nu' \in \theta$ be such that $\beta$ is the only
$\with$ toggled by $\{\nu, \nu'\}$. Again using
Table~\ref{cut-free-induction}, $\nu'=\lambda'\cup\mu'$ for some
linking $\lambda'$ in the proof net of $\Pi'_{21}$. Since there must
be a link $a=\{\ell,\ell'\}$ such that $a \in \lambda$ but
$a\not\in\lambda'$ (or vice versa), in $\fullgraph$ there is a jump
edge from $\ell$ to $\beta$.
This contradicts the assumption that in $\fullgraph$ any path between
(vertices in) $\RvG{\Delta'},A_1$ and $A_2,\Sigma',B_1\parr B_2$ passes through
the indicated occurrence of $\RvG{\alpha}$.

Therefore, by induction, using Lemma~\ref{lem-same}, $\Pi'_{11}$ can be converted
into $\Pi'_{21}$ by a series of rule commutations. Let $\Pi''_2$ be
obtained from $\Pi'_2$ by replacing its subproof $\Pi'_{21}$ by
$\Pi'_{11}$, and let $\Pi'$ be the proof obtained from $\Pi$ by
replacing $\Pi_1$ by $\Pi_1'$ and $\Pi_2$ by $\Pi_2''$.  In
$\Pi'$, the $\RvG{\alpha}$-rules generating $A$ commute with the $\with$-rule $\rho$,
thereby yielding the required proof $\Pi''$. \hfill\filledbox
\end{itemize}

\noindent
\emph{Induction step for Lemma~\ref{lem-sep}.}
Suppose that $\Pi$ does not generate $A=A_1 \alpha A_2$ in its last step.
The case \RvG{$\alpha\in\{\tensor,\cutlabel\}$} is implied by Lemma~\ref{lem-path}.
Therefore we assume here that $\alpha\in\{\plus,\parr,\with\}$.
\begin{itemize}
\item
First consider the case that the last step of $\Pi$ is the
application of a $\mixlabel$-rule $\rho$. Then $\Gamma= \Delta,A$ and
$A$ occurs in a hypothesis $\Delta_d, A$ of $\rho$ (where $\Delta_d
\subseteq \Delta$).  Let $\Pi_d$ be the branch of $\Pi$ above $\rho$ proving
$\Delta_d, A$. Its proof net is separated by $A$ in $\Delta_d$, for
otherwise the proof net $\theta$ of $\Pi$ would not be separated by
$A$ in $\Gamma$.  Thus, by induction, by a series of rule commutations
$\Pi_d$ can be be converted into a proof $\Pi'_d$ of $\Delta_d, A$
that generates $A$ in its last step. Let $\Pi'$ be the proof of
$\Gamma$ obtained by replacing $\Pi_d$ by $\Pi_d'$ in $\Pi$. In
$\Pi'$, $\rho$ commutes over the $\alpha$-rule generating $A$, thereby
yielding the required proof $\Pi''$.
\end{itemize}
In the remaining cases let the last step of $\Pi$ be the application
of a $\beta$-rule $\rho$, generating the formula $B_1 \beta B_2$. Thus
$\Gamma= \Delta, A, B_1 \beta B_2$.
\begin{itemize}
\item 
Let $\beta\in\{\tensor,\plus\RvG{,\cutlabel}\}$. Then $A$ occurs in a hypothesis
$\Delta_d, A, B_d$ of $\rho$ (where $d$ is $1$ or $2$, and
$\Delta_d=\Delta$ in the case $\beta=\plus$).  Let $\Pi_d$ be the
branch of $\Pi$ above $\rho$ proving $\Delta_d, A, B_d$. Its proof net
is separated by $A$ in $\Delta_d, A, B_d$, for otherwise the proof net
$\theta$ of $\Pi$
would not be separated by $A$ in $\Gamma$.  Thus, by induction, by a series of
rule commutations $\Pi_d$ can be be converted into a proof $\Pi'_d$ of
$\Delta_d, A, B_d$ that generates $A$ in its last step. Let $\Pi'$ be
the proof of $\Gamma$ obtained by replacing $\Pi_d$ by $\Pi_d'$ in
$\Pi$. In $\Pi'$, $\rho$ commutes over the $\alpha$-rule generating $A$,
thereby yielding the required proof $\Pi''$.
\item 
Let $\beta = \parr$.
Let $\Pi_\rho$ be the part of $\Pi$ above $\rho$. Then $\Pi_\rho$
proves the hypothesis $\Delta, A, B_1, B_2$ of $\rho$, and its proof
net is separated by $A$, for otherwise $\theta$ would not be separated by
$A$. Thus, by induction, by a series of rule commutations $\Pi_\rho$
can be be converted into a proof $\Pi'_\rho$ of $\Delta, A, B_1, B_2$
that generates $A$ in its last step. As above, a rule commutation
completes the argument.
\item 
Let $\beta = \with$.
Then $\rho$ has hypotheses $\RvG{\Omega_1,\Delta'}, A, B_1$ and
$\RvG{\Omega_2,\Delta'}, A, B_2$ \RvG{with $\Delta = \Omega_1,\Omega_2,\Delta'$}.  Let $\Pi_i$ be the
branch of $\Pi$ above $\rho$ proving $\RvG{\Omega_i,\Delta'}, A, B_i$. The proof nets
of $\Pi_1$ and $\Pi_2$ are separated by $A$ in $\Delta, A, B_i$ in exactly the
same way, \ie, in case $\alpha=\plus$ choosing the same argument
$A_d$, for otherwise $\theta$ would not be
separated by $A$.  By induction, by a series of rule commutations
the $\Pi_i$ can be converted into proofs $\Pi'_i$ of $\RvG{\Omega_i,\Delta'},A,B_i$ that
generate $A$ in their last steps.
Let $\Pi'$ be the proof of $\Gamma$
obtained by replacing $\Pi_i$ by $\Pi_i'$ in $\Pi$, for $i=1,2$. In
$\Pi'$, the $\with$-rule $\rho$ commutes over the $\alpha$-rules
generating $A$, thereby yielding the required proof $\Pi''$.
\hfill\filledbox
\end{itemize}
\vspace{1ex}

\noindent
\emph{Induction step for Lemma~\ref{lem-no-path}.}
\begin{itemize}
\item
First consider the case that the last step $\rho$ of $\Pi$ is an
application of $\mixlabel$, say with hypotheses $\Gamma_1$ and $\Gamma_2$.
Let $\Pi_i$ be the branch of $\Pi$ above $\rho$, proving $\Gamma_i$
(for $i=1,2$). Since $\jumpgraph{\Pi_i}$ is a subgraph of $\jumpgraph\Pi$,
in $\jumpgraph{\Pi_i}$ there is no path between (vertices in) $\Delta_i:=\Delta \cap
\Gamma_i$ and $\Sigma_i := \Sigma\cap\Gamma_i$. In case $\Delta_i$ or
$\Sigma_i$ is empty, we let $\Pi_i'=\Pi_i$. Otherwise, by induction
$\Pi_i$ can, by means of rule commutations, be converted into a
proof $\Pi_i'$ whose last step is a $\mixlabel$-rule
with hypotheses $\Delta_i$ and $\Sigma_i$. Let $\Pi'$ be the proof
obtained from $\Pi$ by replacing $\Pi_i$ with $\Pi_i'$ for $i=1,2$.
In $\Pi'$, $\rho$ commutes over the 0, 1 or 2 $\mixlabel$-rules
introduced immediately above it, thereby yielding the required proof $\Pi''$.
\end{itemize}
In the remaining cases let the last step of $\Pi$ be the application
of a $\beta$-rule $\rho$, generating the formula $B = B_1 \beta B_2$.
We treat the case that $B$ occurs in $\Sigma$;
the other case follows by symmetry. Let $\Sigma = \Sigma',B_1\beta B_2$.
\begin{itemize}
\item 
Let \RvG{$\beta\in\{\tensor,\cutlabel\}$}. The hypotheses of this rule are
$\Delta_i,\Sigma_i,B_i$, for $i \in \{1,2\}$, where
$\Delta=\Delta_1,\Delta_2$ and $\Sigma'=\Sigma_1,\Sigma_2$.
Let $\Pi_i$ be the branch of $\Pi$ proving $\Delta_i,\Sigma_i,B_i$.
Since $\jumpgraph{\Pi_i}$ is a subgraph of $\jumpgraph\Pi$, in
$\jumpgraph{\Pi_i}$ there is no path between (vertices in) $\Delta_i$ and $\Sigma_i,B_i$.
In case $\Delta_i$ is empty, we let $\Pi_i'=\Pi_i$. Otherwise, by induction
$\Pi_i$ can, by means of rule commutations, be converted into a
proof $\Pi_i'$ whose last step is a $\mixlabel$-rule
with hypotheses $\Delta_i$ and $\Sigma_i,B_i$. Let $\Pi'$ be the proof
obtained from $\Pi$ by replacing $\Pi_i$ with $\Pi_i'$ for $i=1,2$.
In $\Pi'$, $\rho$ commutes over the 1 or 2 $\mixlabel$-rules
introduced immediately above it (possibly using $\comm{\mixlabel}{\RvG{\beta}}$
twice and $\comm{\mixlabel}{\mixlabel}$ once), thereby
yielding the required proof $\Pi''$.
\item 
Let $\beta = \plus$. The hypothesis of this rule is
$\Delta,\Sigma',B_d$, where $d$ is 1 or 2. Let $\Pi_d$ be the subproof
of $\Pi$ proving the latter sequent. Since $\jumpgraph{\Pi_d}$ is a
subgraph of $\jumpgraph\Pi$, in $\jumpgraph{\Pi_d}$ there is no path
between (vertices in) $\Delta$ and $\Sigma',B_d$. By induction
$\Pi_d$ can, by means of rule commutations, be converted into a
proof $\Pi_d'$ whose last step is an application of the $\mixlabel$-rule
with hypotheses $\Delta$ and $\Sigma',B_d$. Let $\Pi'$ be the proof
obtained from $\Pi$ by replacing $\Pi_d$ with $\Pi_d'$.
In $\Pi'$, $\rho$ commutes over the $\mixlabel$-rule
introduced immediately above it, thereby yielding the required proof $\Pi''$.
\item 
Let $\beta = \parr$. The hypothesis of this rule is
$\Delta,\Sigma',B_1,B_2$. Let $\Pi_\rho$ be the subproof
of $\Pi$ proving the latter sequent. Since $\jumpgraph{\Pi_\rho}$ is a
subgraph of $\jumpgraph\Pi$, in $\jumpgraph{\Pi_\rho}$ there is no path
between $\Delta$ and $\Sigma',B_1,B_2$. By induction
$\Pi_\rho$ can, by means of rule commutations, be converted into a
proof $\Pi_\rho'$ whose last step is a $\mixlabel$-rule
with hypotheses $\Delta$ and $\Sigma',B_1,B_2$. Let $\Pi'$ be the
proof obtained from $\Pi$ by replacing $\Pi_\rho$ with $\Pi_\rho'$.
In $\Pi'$, $\rho$ commutes over the $\mixlabel$-rule
introduced immediately above it, thereby yielding the required proof $\Pi''$.
\item 
Let $\beta = \with$. The hypotheses of this rule are
$\RvG{\Omega^\Delta_i,\Delta'}, \RvG{\Omega^\Sigma_i,\Sigma''}, B_i$
for $i \mathbin\in \{1,2\}$ \RvG{with $\Delta = \Omega^\Delta_1,\Omega^\Delta_2,\Delta'$
and $\Sigma' = \Omega^\Sigma_1,\Omega^\Sigma_2,\Sigma''$}.
As in the induction step for Lemma~\ref{lem-path}, it follows that
$\Omega_1$ and $\Omega_2$ are empty.
Let $\Pi_i$ be the branch of $\Pi$ proving $\RvG{\Delta'},\RvG{\Omega^\Sigma_i,\Sigma''},B_i$.
Since $\jumpgraph{\Pi_i}$ is a subgraph of $\jumpgraph\Pi$, in
$\jumpgraph{\Pi_i}$ there is no path between (vertices in) $\Delta'$ and $\RvG{\Omega^\Sigma_i,\Sigma''},B_i$.
By induction $\Pi_i$ can, by means of rule commutations, be converted into a
proof $\Pi_i'$ whose last step is a $\mixlabel$-rule
with hypotheses $\Delta'$ and $\RvG{\Omega^\Sigma_i,\Sigma''},B_i$.
So the left hypotheses of $\Pi'_1$ and $\Pi'_2$
are both $\Delta'$, and we claim that the proof nets on them induced
by the subproofs $\Pi'_{11}$ and $\Pi'_{21}$ of $\Pi$ leading up to these
hypotheses must be the same. The argument goes just as in the
induction step for Lemma~\ref{lem-path}.

Therefore, by induction, using Lemma~\ref{lem-same}, $\Pi'_{11}$ can be converted
into $\Pi'_{21}$ by a series of rule commutations. Let $\Pi''_2$ be
obtained from $\Pi'_2$ by replacing its subproof $\Pi'_{21}$ by
$\Pi'_{11}$, and let $\Pi'$ be the proof obtained from $\Pi$ by
replacing $\Pi_1$ by $\Pi_1'$ and $\Pi_2$ by $\Pi_2''$.  In
$\Pi'$, the $\mixlabel$-rules generating $A$ commute with the
$\with$-rule $\rho$, 
thereby yielding the required proof $\Pi''$. \hfill\filledbox
\end{itemize}

\vspace{1ex}

\noindent
\emph{Induction step for Lemma~\ref{lem-same}.}   For the induction step, suppose
$\Pi$ and $\Pi'$ are two proofs of a \RvG{\mallcut} sequent $\Gamma$
that have the same proof net $\theta$. 

First assume that $\fullgraph$ is connected.  In that case the last
steps of $\Pi$ and $\Pi'$ cannot be $\mixlabel$. Let $A$ be the
formula occurrence in $\Gamma$ that is generated by the last step of
$\Pi'$.  By Lemma~\ref{lem-last-gen}, $A$ separates $\theta$. Hence,
using Lemma~\ref{lem-sep}, by means of a series of rule commutations, $\Pi$ can
be converted into a proof $\Pi''$ of $\Gamma$ that generates $A$ in
its last step. By Lemma~\ref{lem-conn}, the last step $\sigma$ of $\Pi'$ is the
same as the last step of $\Pi''$.  Thus each hypothesis $\Gamma_d$ of
$\sigma$ is proven by a subproof $\Pi'_d$ of $\Pi'$, and by a subproof
$\Pi''_d$ of $\Pi''$.  As $\Pi'_d$ and $\Pi''_d$ have the same proof
net, by induction they can be converted into each other by means of a
series of rule commutations. It follows that also $\Pi$ and $\Pi'$
can be converted into each other by means of a series of rule
commutations.

Next assume that $\fullgraph$ is disconnected; let $\Gamma =
\Gamma_1,\Gamma_2$ with the $\Gamma_i$ nonempty sequents,
such that in $\fullgraph$ there is no path between (vertices in)
$\Gamma_1$ and $\Gamma_2$. Using Lemma~\ref{lem-no-path}, $\Pi$ can, by means of rule
commutations, be converted into a proof $\Pi_\mixlabel$ whose last step is the
$\mixlabel$-rule with hypotheses $\Gamma_i$. Let $\Pi_i$ be the branch
of $\Pi_\mixlabel$ proving $\Gamma_i$. Its proof net is simply the
restriction of (the linkings in) $\theta$ to $\Gamma_i$.
Likewise, $\Pi'$ can, by means of rule
commutations, be converted into a proof $\Pi'_\mixlabel$ whose last step is the
$\mixlabel$-rule with hypotheses $\Gamma_i$. Let $\Pi'_i$ be the branch
of $\Pi_\mixlabel$ proving $\Gamma_i$. Since $\Pi_i$ and $\Pi'_i$ have
the same proof net, by induction one can be converted into the other
by a series of rule commutations.
Consequently, $\Pi$ can be converted into $\Pi'$.
\end{proof}

\section{Proof of the MALL rule commutation theorem}

We use Proposition~\ref{MALL* kernel} to derive Theorem~\ref{kernel
  cut}. We shall need two lemmas connecting {\mallcut} rule
commutations with \fullMALL rule commutations.
\begin{lemma}\label{lem comm 1}
If two {\mallcut} proofs differ by a rule commutation, so do
their projections to \fullMALL proofs.
\end{lemma}
\begin{proof}
This follows immediately from inspecting the rule commutations.
\end{proof}
In the other direction, one might expect that for each pair
$(\Pi_l,\Pi_r)$ of commuting \fullMALL proofs, and for each {\mallcut}
proof $\Pi^*_l$ that projects to $\Pi_l$, there exists a {\mallcut} proof $\Pi^*_r$
projecting to $\Pi_r$ and commuting with $\Pi^*_l$. This is not the
case, however. A counterexample is provided by taking $\Pi^*_r$ to be
\begin{center}
\vspace{12pt}
\renewcommand{\piproof}[2]{\[\;\Pi_{#1}^*\justifies{#2}\]}
\begin{prooftree}\thickness=.08em
\[
 \piproof{1a}{\Omega_1,\Gamma,\A} \piproof{2}{\B,\Delta,\X} \justifies
 \Omega_1,\Gamma,\A\tensor \B,\Delta,\X \using \tensorlabel
\]
\[
 \piproof{1b}{\Omega_2,\Gamma,\A} \piproof{3}{\B,\Delta,\Y} \justifies
 \Omega_2,\Gamma,\A\tensor \B,\Delta,\Y \using \tensorlabel
\]
\justifies
\Omega_1,\Omega_2,\Gamma,\A\tensor \B,\Delta,\X\with \Y
\using \withlabel
\end{prooftree}
\vspace{12pt}
\end{center}
with $\Omega_1$ or $\Omega_2$ nonempty, and $(\Pi_l,\Pi_r)$ the last
rule commutation of Table~\ref{nondiag-commutations}.  For
$(\Pi_l,\Pi_r)$ to be a valid rule commutation, the subproofs
$\Pi_{1a}^*$ and $\Pi_{1b}^*$ must project to identical \fullMALL
proofs, even though they derive different {\mallcut}
sequents. This can be achieved by inserting a
$\with$-rule in each of these subproofs, where one superimposes two
cuts, while the other keeps them disjoint.
However, a weaker property \emph{does} hold:
\begin{lemma}\label{lem comm 2}
For each \fullMALL rule commutation $(\Pi_l,\Pi_r)$ there is a
{\mallcut} rule commutation $(\Pi^*_l,\Pi^*_r)$ that projects to
$(\Pi_l,\Pi_r)$.
\end{lemma}
\begin{proof}
Orient the pair $(\Pi_l,\Pi_r)$ so that we avoid $\Pi_l$ being an
$\alpha\beta$-proof fragment (see the Appendix)
with $\beta=\with$ and $\alpha\in\{\tensor,\mixlabel,{\sf cut}\}$.
Take $\Pi^*_l$ to be an arbitrary {\mallcut} proof projecting to $\Pi_l$.
Going through the rule commutations of
Tables~\ref{diag-commutations}--\ref{cut-commutations},
one can check that in each case it is straightforward to find the
required proof $\Pi^*_r$.
\end{proof}

\begin{corollary}\mbox{}\vspace{-1ex}\label{corr-kernel cut}
\begin{enumerate}
\item[(a)]
If two \fullMALL proofs\/ $\Pi_l$ and $\Pi_r$ translate to a common proof net
then they can be converted into each other by rule commutations.
\vspace{-1ex}
\item[(b)]
If two \fullMALL proofs\/ $\Pi_l$ and $\Pi_r$ differ by a rule commutation then they have a common proof net.
\end{enumerate}
\end{corollary}

\begin{proof}
  Suppose $\Pi_l$ and $\Pi_r$ translate to a common proof net $\theta$.
  Then $\Pi_l$ and $\Pi_r$ must be projections of {\mallcut} proofs
  $\Pi^*_l$ and $\Pi^*_r$ that translate to $\theta$. By
  Proposition~\ref{MALL* kernel} $\Pi^*_l$ and $\Pi^*_r$ can be
  converted into each other by a series of rule commutations.
  By Lemma~\ref{lem comm 1} the same holds for $\Pi_l$ and $\Pi_r$.

  Suppose $\Pi_l$ and $\Pi_r$ differ by a rule commutation.
  By Lemma~\ref{lem comm 2} there are {\mallcut} proofs $\Pi^*_l$ and
  $\Pi^*_r$ that differ by a rule commutation and project to $\Pi_l$
  and $\Pi_r$. By Proposition~\ref{MALL* kernel} $\Pi^*_l$ and
  $\Pi^*_r$ translate to the same proof net $\theta$. Hence $\theta$
  is a common proof net of $\Pi_l$ and $\Pi_r$.
\end{proof}
Finally, Theorem~\ref{kernel cut} is a direct consequence of
Corollary~\ref{corr-kernel cut}.\hfill$\square$

\section{Alternative treatments of cut}

One of the innovations of the proof nets from \cite{HvG} over the
monomial ones from \cite{Gir87} is that the translation from cut-free
proofs to proof nets is a function. This property does not extend to
proofs with cut. In \cite[Section 5.3.4]{HvG} \RvG{three} alternative
translations are discussed \RvG{of which two} \emph{are} functions.
\RvG{One of these fails to identify} proof nets modulo rule commutations.
\RvG{For the other, we conjecture that it does.
However, for this notion ``it is not immediately clear how to define a meaningful correctness
criterion to characterise the image of the translation'' \cite{HvG}.}

\paragraph{Superimposing no cuts}
The first alternative is to restrict the rule for $\with$ in
Table~\ref{mallcut-rules} by requiring that $\Gamma$ may contain no
cuts. This means the cuts appearing in the conclusion of the rule must
be the disjoint union of the cuts appearing in the premises.
Let \fullMALL{}$^{\mkern-5mu\cutlabel}_{\sf sep}$ be the resulting alternative for {\mallcut}.
Now each \fullMALL proof is the restriction of a unique \fullMALL{}$^{\mkern-5mu\cutlabel}_{\sf sep}$
proof; hence the translation from \fullMALL proofs to proof nets becomes a
function.

Clearly, the resulting notion of proof-net equivalence on \fullMALL proofs
is included in the one from Section~\ref{cut-sec}. In fact the
conclusion is strict, for we loose the rule commutation
$\comm{\cutlabel}{\with}$, as illustrated in \cite[Section 5.3.4]{HvG}.
\RvG{In general}, the commutations $\comm{\alpha}{\with}$ and
$\comm{\with}{\alpha}$ with $\alpha\in\{\tensor,\mixlabel,\cutlabel\}$
are no longer valid, because any cut included in $\Gamma$ appears only
once on the left, yet is duplicated on the right.
(All other rule commutations remain valid.)

\paragraph{Superimposing as many cuts as possible}
The alternative of requiring $\Omega_1$ and $\Omega_2$ in the rule for $\with$ in
Table~\ref{mallcut-rules} to be disjoint superimposes as many cuts as possible.
As pointed out in \cite[Section 5.3.4]{HvG} it does not yield a
function from \fullMALL proofs to proof nets, for there may be a choice of
how to identify cuts.

\paragraph{Local cuts}
A final variation considered in \cite{HvG} is to depart
from sets of linkings on a fixed cut sequent, and permit each linking
its own set of cut pairs. Define a \defn{cut linking} on a \fullMALL sequent
$\Gamma$ as a linking on a sequent $\Omega,\Gamma$ with $\Omega$ a
disjoint union of cuts. In order to abstract from the identity of the cut pairs we
consider $\Omega$ (but not $\Gamma$) up to isomorphism. A \fullMALL proof
of $\Gamma$ yields a set of cut linkings on $\Gamma$ in the obvious way \cite{HvG}.
This yields a deterministic translation (function) from \fullMALL proofs to
sets of cut linkings.

Since the set of cut linkings of a \fullMALL sequent $\Gamma$ can be
inferred from any \fullMALL proof net of $\Gamma$, the kernel of this
function (identifying \fullMALL proofs that translate to the same set of
cut linkings) \RvG{includes} proof-net equivalence as defined in
Section~\ref{cut-sec}.
\RvG{Thus, two MALL proofs that differ by rule commutations translate
  to the same set of cut linkings.}

\begin{conjecture}
Two \fullMALL proofs translate to the same set of cut linkings if and only
if they can be converted into each other by a series of rule commutations.
\end{conjecture}

\section{Local rule commutations}\label{sec:local}

The rule commutations 
$\commpair{\with}{\tensor}$,
$\commpair{\with}{\mixlabel}$,
$\commpair{\with}{\truecutlabel}$
and
$\commpair{\with}{\cutinconv}$
duplicate/identify premises, respectively;
we refer to the other rule commutation as \defn{local}.
The Appendix below concludes with a general definition of local rule commutation.
In \cite{HH16} a different notion of proof net, called a
\defn{conflict net}, is proposed, such that two
\fullMALL proofs translate to the same conflict net if and only
if they can be converted into each other by a series of local rule commutations.

\section*{Appendix: \kern3pt General concept of rule commutation}

In order to properly define rule commutations in a sequent calculus,
we consider rules---called \defn{abstract} rules---that contain
\defn{variables} ranging over formulas and over sequents. The rules
for \fullMALL in Section\RvG{s}~\ref{mall-sec} \RvG{and~\ref{cut-sec}}
are of this form. Thus, rather than seeing the
rule for $\tensorlabel$ as a template, of which there is an instance
for each choice of $A$, $B$, $\Gamma$ and $\Delta$, we see it as a
single rule containing four variables.  When applying such a rule in a
proof, formulas and sequents are substituted for the variables of the corresponding type.

Formally, a \defn{formula expression} is built from formula variables,
negated formula variables, literals and connectives; it is a
\defn{formula} if it contains only literals and connectives.
Here a \defn{negated formula variable} is a formula variable annotated
with the subscript $\perp$. A \defn{sequent expression} is a multiset of sequent
variables and formula expressions; it is a \defn{sequent} if it does not
contain any variables. Here a \defn{multiset} of objects from a set
$S$ is a function $M:S\rightarrow\mbox{I\!N}$ indicating for each object in
$S$ how often it occurs in $M$. An object $x\in S$ with $M(x)>0$ is
called an \defn{element} of $M$. Let $C(M)=\{x\in S \mid M(x)>0\}$
denote the set of elements of $M$. In case $M(x)\in \{0,1\}$ for all
$x\in S$, the multiset $M$ is usually identified with the set $C(M)$.

An \defn{abstract rule} is a pair \plat{$\frac{H}{\Gamma}$} of a
set $H$ of sequent expressions---the \defn{premises}---and a single
sequent expression $\Gamma$---the conclusion.
A \defn{concrete rule}---simply called \defn{rule} outside of this
appendix---is a pair \plat{$\frac{H}{\Gamma}$} of a \emph{multi}set
$H$ of (variable-free) sequents and a single sequent $\Gamma$.

A \defn{substitution} $\sigma$ maps formula variables to formula
expressions and sequent variables to sequent expressions;\footnote{\RvG{In order
to capture {\mallcut}, we also allow sequent variables of special
types---like ``cut only'' in Table~\ref{mallcut-rules}---and for each
type define the class of sequent expressions that may be substituted
for it.}} it extends
to negated formula variables $A\perp$ by $\sigma(A\perp)=\sigma(A)\perp$,
and further extends to a map from formula expressions to formula expressions and from
(sets of) sequent expressions to (multisets of) sequent
expressions. A substitution is \defn{closed} if it maps formula
variables to formulas and sequent variables to sequents.
If \plat{$\frac{H}{\Gamma}$} is an abstract rule and $\sigma$ a
(closed) substitution, then \plat{$\frac{\sigma(H)}{\sigma(\Gamma)}$}
is a (closed) \defn{substitution instance} of \plat{$\frac{H}{\Gamma}$};
its \defn{collapse} \plat{$\frac{C(\sigma(H))}{\sigma(\Gamma)}$}
is again an abstract rule.

Given a collection of connectives to determine the valid formulas,
a \defn{sequent calculus}---such as \fullMALL---is given by a set of abstract rules.%
\footnote{By these definitions, the \fullMALL axiom $\axlabel$, unlike the
  other rules, is still a template, of which an instance is obtained
  by filling in actual propositional variables for the metavariable $P$.
  If this is felt to be inelegant, one could rename ``propositional
  variable'' into ``atom'' and introduce ``atom variables'' and
  negated atom variables to formulate the axiom $\axlabel$.
  For simplicity, we abstain from doing this here.}
It induces a set of concrete rules, namely the \RvG{collapsed} closed substitution
instances of the abstract rules.

We now formalise proofs, extended to include the case where the
conclusion is a sequent expression. When the conclusion is a standard
sequent, the definition specialises to the familiar notion of sequent
calculus proof.  A \defn{proof} $\Pi$ in a sequent calculus is a
well-founded, upwards branching tree whose nodes are labelled by
sequent expressions and some of the leaves are marked ``hypothesis'',
such that
if $\Delta$ is the label of a node that is not a hypothesis
and $K$ is the multiset of labels of the children of this node then
$\frac{K}{\Delta}$ is a substitution instance of one of the rules of
that sequent calculus.
Such a proof \defn{derives} the abstract rule $\frac{H}{\Gamma}$, where $H$ is
the set of labels of the hypotheses, and $\Gamma$ the label of the
root of $\Pi$. A proof of a sequent expression $\Gamma$ can be
regarded as a proof of the abstract rule $\frac{H}{\Gamma}$ with $H=\emptyset$.

For $\alpha$ and $\beta$ two abstract proof rules in a sequent calculus,
an \defn{$\alpha\beta$-proof} is a proof in which each non-hypothesis
node is either the root and an application of $\beta$, or a child of the root and an application of $\alpha$.
A \defn{subproof} $\Pi'$ of a proof $\Pi$ comprises all nodes in the tree
$\Pi$ above a given node, which is the root of $\Pi'$.
A proof $\Pi_f$ deriving a rule $\frac{H}{\Gamma}$, together with
proofs $\Pi_\Delta$ of $\Delta$ for each $\Delta\in H$, composes into
a proof $\Pi'$ of $\Gamma$, such that the proofs $\Pi_\Delta$ are
subproofs of $\Pi'$. If $\Pi'$ itself is a subproof of a proof
$\Pi$ we say that $\Pi_f$ is a \defn{proof fragment} of $\Pi$;
if $\Pi_f$ is an $\alpha\beta$-proof, it is called an 
\defn{$\alpha\beta$-proof fragment} of $\Pi$.\linebreak[2]
For $\Pi$ a proof and $\sigma$ a substitution, $\sigma(\Pi)$ denotes
the proof obtained from $\Pi$ by applying $\sigma$ to all its node labels.

An abstract rule is \defn{pure} if
(1) its premises are free of literals and
connectives and thus are built from variables (sequent variables,
formula variables and negated formula variables) only, and
(2) each of these variables occurs exactly once in the conclusion.
We define rule commutation for sequent calculi containing pure
rules only. This includes \cutfreeMALL and {\mallcut}, but not \fullMALL;
however, the rule commutations of \fullMALL can be derived as the projections
of the ones for {\mallcut}.

The implicit tracking of subformula occurrences described in
Section~\ref{mall-sec} and utilised in Sections~\ref{sec-proofnets}
and~\ref{cut-sec} can now be formalised as follows: a 
subformula occurrence within an occurrence of a formula or sequent substituted
for a variable $A$, $A\perp$ or $\Gamma$ appearing in the premises of
an abstract rule tracks 
to the corresponding subformula occurrence within the occurrence of
the same formula or sequent substituted for $A$, $A\perp$ or
$\Gamma$ in the conclusion of the rule.

It is not hard to show that any abstract rule derivable in a sequent calculus
containing pure rules only can be obtained as a collapsed substitution
instance of a pure rule derivable in that sequent calculus.  Although
we do not make use of this insight in our proofs, it helps to motivate
the following definition.

A \defn{rule commutation} is an (ordered) pair of an
$\alpha\beta$-proof and a (different) $\beta\alpha$-proof
deriving the same pure rule.
An $\alpha\beta$-proof $\Pi'_1$ \defn{commutes} with a
$\beta\alpha$-proof $\Pi'_2$ if there exists a rule commutation
$(\Pi_1,\Pi_2)$ and a substitution $\sigma$ such that $\sigma(\Pi_1)=\Pi'_1$
and $\sigma(\Pi_2)=\Pi'_2$.
Two proofs \defn{differ by a rule commutation} if one can be obtained from
the other by the replacement of an $\alpha\beta$-proof fragment
occurring in it by a commuting $\beta\alpha$-proof fragment.
Thus a rule commutation is a transposition of adjacent rules that
preserves subproofs immediately above, with possible
duplication/identification.

We leave it to the reader to check that this definition, applied to
\cutfreeMALL and {\mallcut}, generates exactly the rule commutations presented in
Sections~\ref{sec-commutations} and~\ref{star commutations}.

In our definition of rule commutation it is essential that the rule
derived by each of the two proofs $\Pi_1$ and $\Pi_2$ in a rule
commutation is pure. Skipping this requirement would give rise to
unwanted rule commutations. As an example, consider the rule
commutation $\comm{\convparr}{\convparr}$ of Figure~\ref{diag-commutations}
in which $\Gamma,A_1$ is substituted for~$\Gamma$.
The two sides of the resulting rule commutation define the same
non-pure rule. Simply requiring---as we do---that the same proof
$\Pi$, deriving the sequent $\Gamma,A_1,A_1,A_2,B_1,B_2$, is used at
both sides of the commutation does not rule out that the roles of
the two occurrences of $A_1$ are swapped at one side of the
commutation, possibly leading to proofs inducing different proof nets.

Moreover, we cannot drop the requirement that $\Pi_1$ and $\Pi_2$ must
be $\alpha\beta$- and $\beta\alpha$-rules, for that would give rise
to the unwanted rule commutation
\begin{center}
  \vspace{2ex}\hspace*{-2ex}\begin{math}
    \newcommand{\gap}{\hspace{5.5ex}}
    \newcommand{\branch}[3]{
      \[
        \[
          #1,\, C \hspace{4ex} D,\ #2
          \justifies
          #1,\, C\tensor D,\ #2
          \using \tensorlabel
        \]
        \justifies
        #1,\, C\tensor D,\, E\plus F
        \using #3
      \]
    }
    \newcommand{\base}[4]{
      \begin{prooftree}\thickness=.08em
        \branch{A}{#1}{#3}
        \hspace{1ex}
        \branch{B}{#2}{#4}
        \justifies
        A \with B,\, C\tensor D,\, E\plus F
        \using \withlabel
      \end{prooftree}
    }
    \begin{array}{c@{\gap}c@{\gap}c}
      \base{E}{F}{\leftpluslabel}{\rightpluslabel}
      & \longleftrightarrow &
      \base{F}{E}{\rightpluslabel}{\leftpluslabel}
      \\[8ex]
    \end{array}
  \end{math}\hspace*{-2ex}
\end{center}
These two proofs derive the same pure rule, yet (when instantiated)
induce different proof nets:
\begin{center}\vspace{2ex}%
  \begin{math}
    \newcommand{\Ax}{4}
    \newcommand{\Bx}{19}
    \newcommand{\Cx}{39}
    \newcommand{\Dx}{58}
    \newcommand{\Ex}{78}
    \newcommand{\Fx}{97}
    \newcommand{\net}[2]{
      \links{A \with B,\;\;\;\: C\tensor D,\;\;\;\: E\plus F\rule{0ex}{1.8ex}}{%
      \renewcommand{\linkcolor}{\Brown}\olink{\Ax}{\Cx}%
      \renewcommand{\linkcolor}{\Blue}\olink{\Dx}{#1}%
      \renewcommand{\linkcolor}{\Red}\ulink{\Bx}{\Cx}%
      \renewcommand{\linkcolor}{\Green}\ulink{\Dx}{#2}}
    }
    \net{\Ex}{\Fx}
    \hspace*{15ex}
    \net{\Fx}{\Ex}
  \end{math}\vspace{2ex}
\end{center}
%
Based on the above, we say that a concrete $\beta$-rule \defn{commutes over}
a concrete $\alpha$-rule, if these rules occur in an $\alpha\beta$-proof fragment
obtained as a substitution instance of an $\alpha\beta$-proof for which
there exists a $\beta\alpha$-proof deriving the same pure rule.
This definition of rule commutation is more liberal than the
standard definition of rule commutation for a Gentzen sequent calculus
\cite[Def.\ 5.2.1]{TS96}, analysed by Kleene \cite{Kle52} and Curry
\cite{Cur52}. That definition only covers the case where \emph{each}
$\beta$ rule commutes over \emph{each} $\alpha$-rule, corresponding
with the check marks in Table~\ref{commutation-table}.
Moreover, \cite{TS96} requires---translated to our terminology---the
source proof fragment to have two non-leaf nodes only (one for $\beta$
and only one for $\alpha$), thereby ruling out the commutation of
$\with$ over any $\alpha$.

\paragraph{Local rule commutations.}

Define a proof as \defn{non-repeating} if all its hypothesis have a
different label.
A rule commutation $(\Pi_1,\Pi_2)$ is \defn{local} (\cf\
Section~\ref{sec:local}) if $\Pi_1$ and $\Pi_2$ are non-repeating.

\bibliographystyle{eptcs}

\end{document}